\documentclass[11pt]{article}
\usepackage{microtype}
\usepackage[T1]{fontenc}
\usepackage[mathscr]{euscript}
\usepackage[bb=stixtwo]{mathalpha}
\usepackage{eulervm}
\usepackage[light]{CormorantGaramond}
\usepackage[T1]{fontenc}
\usepackage{graphicx}

\usepackage{amssymb}
\usepackage{amsmath}
\usepackage{amsthm}
\usepackage{amssymb}

\usepackage{bm}
\usepackage{amsfonts}

\usepackage{blkarray}

\usepackage{float}

\usepackage{url} 

\usepackage[dvipsnames]{xcolor}

\definecolor{violet}{RGB}{124,35,217}

\usepackage[bookmarks,bookmarksopen,bookmarksnumbered, colorlinks=true,citecolor=violet]{hyperref}
\usepackage{xspace}
\usepackage{circledsteps}
\usepackage{algorithm}
\usepackage[commentColor=ForestGreen,italicComments=true,noEnd=false]{algpseudocodex}
\usepackage{booktabs,ctable,multirow}
\usepackage[english,capitalise]{cleveref}
\newtheorem{corollary}{Corollary}

\newtheorem{lemma}{Lemma}

\newtheorem{remark}{Remark}

\DeclareMathOperator{\theargmin}{argmin}
\DeclareMathOperator{\theargmax}{argmax}

\DeclareMathOperator{\diagonal}{diag}

\DeclareMathOperator{\id}{Id}

\DeclareMathOperator{\pr}{\mathbb{P}}
\DeclareMathOperator{\R}{\mathbb{R}}
\DeclareMathOperator{\OG}{O}

\DeclareMathOperator{\support}{supp}
\DeclareMathOperator{\thesbm}{SBM}

\newcommand{\bA}{\bm{A}}

\newcommand{\hA}{{\widehat{\bm{A}}}}

\newcommand{\bAt}{\bm{A}^{(t)}}
\newcommand{\mA}{\sE{\pr}}

\newcommand{\bB}{\bm{B}}

\newcommand{\bD}{\bm{D}}

\newcommand{\hD}{{\widehat{\bm{D}}}}
\newcommand{\Dbar}{\overline{\bD}}
\newcommand{\dbar}{\overline{d\mspace{4mu}}}

\newcommand{\bE}{\bm{E}}

\newcommand{\eps}{\varepsilon}
\newcommand{\eqdef}{\stackrel{\text{\tiny def}}{=}}

\newcommand{\citeg}[1]{(e.g., \cite{#1}, and references therein)}

\newcommand{\fr}{Fr\'echet\xspace}

\newcommand{\Gt}{G^{(t)}}

\newcommand{\ist}{{i^\ast}}

\newcommand{\bJ}{\bm{J}}

\newcommand{\bLambda}{\bm{\Lambda}}
\newcommand{\blb}{\bm{\lambda}}

\newcommand{\hlm}{\widehat{\lambda}}
\newcommand{\bL}{\bm{L}}
\newcommand{\cL}{\bm{\mathcal{L}}}
\newcommand{\cLbar}{{\overline{\bm{\mathcal{L}}}}\mspace{4mu}}
\newcommand{\cLt}{\bm{\mathcal{L}}^{(t)}}

\newcommand{\hLM}{\widehat{\bm{\mathcal{L}}}_M \big(\fmpr\big)}
\newcommand{\spcL}{\lmb{\cL(\fmpr)}} 
\newcommand{\Lmu}{{\cL\big(\sfm{\pr}\big)}}

\newcommand{\mpr}{{m^\prime}}

\newcommand{\fmpr}{{\sfm{\pr}}}
\newcommand{\fmprM}{{\widehat{\bm{\mu}}}_T^M\big[\pr\big]}

\newcommand{\one}{\bm{1}}

\newcommand{\bpsi}{\bm{\psi}}
\newcommand{\bPsi}{\bm{\Psi}}
\newcommand{\bp}{\bm{p}}
\newcommand{\bP}{\bm{P}}

\newcommand{\bQ}{\bm{Q}}

\newcommand{\cR}{\mathcal{R}}

\newcommand{\cS}{\mathcal{S}}

\newcommand{\zr}{\bm{0}}

\NewDocumentCommand \F {o m}   
{
  \IfNoValueTF {#2} {\widehat{F}\left(#2\right)}{\widehat{F}_{#1}\mspace{-2mu}\left(#2\right) }
}

\DeclareRobustCommand{\argmin}[1]{\underset{#1}{\theargmin}\mspace{4mu}}
\DeclareRobustCommand{\argmax}[1]{\underset{#1}{\theargmax}\mspace{4mu}}

\DeclareRobustCommand{\diag}[1]{\mspace{4mu}{\diagonal\big(#1 \mspace{-2mu}\big)}\mspace{14mu}}

\DeclareRobustCommand{\E}[1]{{\mathbb{E} \mspace{-2mu}\big[\mspace{-1mu} #1 \mspace{-1mu} \big]}}             

\DeclareRobustCommand{\ipr}[2]{{\langle #1, #2 \rangle}} 

\newcommand{\lmb}[1]{{\bm{\lambda} \big( #1 \big)}}
\DeclareRobustCommand{\lon}[1]{{\lvert #1 \rvert}}

\DeclareRobustCommand{\norm}[1]{{\big \|#1 \big \|}}
\DeclareRobustCommand{\mse}[1]{{n^{-2}\big \|#1 \big \|^2_{F}}}
\DeclareRobustCommand{\nrm}[2]{{\big \|#1 \big \|_{#2}}}
\DeclareRobustCommand{\o}[1]{{\scriptstyle\mathscr{O}} \textstyle \left(\displaystyle #1\right)}
\DeclareRobustCommand{\O}[1]{{\displaystyle \mathcal{O}} \mspace{-4mu} \left(#1\right)}

\DeclareRobustCommand{\prob}[1]{\mspace{2mu}\mathbb{P}\mspace{-2mu}\left(\mspace{-2mu} #1 \mspace{-2mu}\right)}

\DeclareRobustCommand{\rstr}[2]{{#1\raisebox{-.5ex}{\big\vert}_{#2}}}

\DeclareRobustCommand{\sE}[1]{\widehat{\mathbb{E}}_T\mspace{-2mu}\left[#1 \right]} 
\DeclareRobustCommand{\sfm}[1]{\widehat{\bm{\mu}}_T   \mspace{-2mu}\big [\mspace{-2mu}#1 \mspace{-2mu}\big]}    

\DeclareRobustCommand{\sbm}[3]{{\thesbm \big(#1, #2, #3 \big)}}

\DeclareRobustCommand{\supp}[1]{{\support \big(#1\big)}}

\begin{document}

\title{The Spectral Barycentre of a Set of Graphs\\ with Community Structure}
\author{Fran\c{c}ois G. Meyer\footnote{FGM was supported in part by the National Science Foundation (CCF/CIF 1815971).}\\
  Applied Mathematics\\University of Colorado at Boulder, Boulder CO 80305\\
  \href{mailto:fmeyer@colorado.edu}{\sf \small fmeyer@colorado.edu}\\{\small
    \url{https://francoismeyer.github.io}}
  }
\maketitle
\begin{abstract}
  
     The notion of barycentre graph is of crucial importance for machine
    learning algorithms that process graph-valued data. The barycentre graph is a ``summary
    graph'' that captures the mean topology and connectivity structure of a training dataset
    of graphs.  The construction of a barycentre requires the definition of a metric to
    quantify distances between pairs of graphs. In this work, we use a multiscale spectral
    distance that is defined using the eigenvalues of the normalized graph Laplacian. The
    eigenvalues -- but not the eigenvectors -- of the normalized Laplacian of the barycentre
    graph can be determined from the optimization problem that defines the barycentre. In
    this work, we propose a structural constraint on the eigenvectors of the normalized
    graph Laplacian of the barycentre graph that guarantees that the barycentre inherits the
    topological structure of the graphs in the sample dataset. The eigenvectors can be
    computed using an algorithm that explores the large library of Soules bases. When the
    graphs are random realizations of a balanced stochastic block model, then our algorithm
    returns a barycentre that converges asymptotically (in the limit of large graph size)
    almost-surely to the population mean of the graphs. We perform Monte Carlo simulations
    to validate the theoretical properties of the estimator; we conduct experiments on
    real-life graphs that suggest that our approach works beyond the controlled environment
    of stochastic block models.\\

\end{abstract}
\section{Introduction, problem statement, and related work}
\subsection{The barycentre graph}
The design of machine learning algorithms that can analyze "graph-valued random variables"
is of fundamental importance
\citeg{avrachenkov22,balan22,dubey20,ghoshdastidar20,haasler24,haghani22,kolaczyk20,lunagomez21,petersen19,zambon19}.
Such algorithms often require the computation of a "sample mean" graph that can summarize
the topology and connectivity of a training dataset of graphs, $\big\{G^{(1)},\ldots,
G^{(T)} \big\}$.

Formally, we denote by $\cS$ the set of $n \times n$ symmetric adjacency matrices with
nonnegative weights.

  We equip $\cS$ with a probability $\pr$. An example of a probability space $(\cS,\pr)$ is the
  stochastic block model (SBM), described in \cref{SBM-section}.  

The adjacency matrix $\bAt$ of the graph $\Gt, t=1,\ldots,T$ is sampled from $(\cS,\pr)$.  We equip
$(\cS,\pr)$ with a metric $d$ to quantify proximity of graphs. A notion of summary
graph is provided by the concept of {\em barycentre} \cite{sturm03}, or {\em \fr mean}
graph \cite{blanchard25}, $\fmpr$, which minimizes the sum of the squared distances to all
the graphs in the sample,
\begin{equation}
  \fmpr \eqdef \argmin{\bB\in \cS} \sum_{t=1}^T d^2(\bB,\bAt).
  \label{le-barycentre}
\end{equation}
Before continuing, we introduce some notations. We denote by $[n] \eqdef \{1,\ldots,n\}$. We
define $\one \eqdef [1 \cdots 1]^T$, and $\bJ = \one\one^T$. We use $\bA$ to denote the
adjacency matrix of a graph $G$, and $\bD\eqdef \diag{\bA\one}$ to denote the diagonal
degree matrix. The symmetric normalized adjacency matrix, $\hA = \bD^{-1/2} \bA \bD^{-1/2}$,
is defined by
\begin{equation}
  \hat{a}_{ij} \eqdef a_{ij}/\sqrt{d_id_j} \mspace{8mu}\text{if}\mspace{8mu} d_id_j\neq
  0; \mspace{8mu}\text{and} \mspace{8mu} \hat{a}_{ij}  \eqdef 0\mspace{8mu}
  \text{otherwise}.
  \label{Ahat}
\end{equation}
The normalized Laplacian is defined by $\cL \eqdef \id - \hA$. We denote
by $\blb(\cL) = [\lambda_1,\ldots,\lambda_n]$ the ascending sequence of
eigenvalues of $\cL$, $0 = \lambda_1 \le \cdots \le \lambda_n \le 2$.

\subsection{The Laplacian spectral pseudo-distance}
The metric $d$ in \cref{le-barycentre} influences the topological characteristics that
$\sfm{\mu}$ inherits from $\{G^{(1)},\ldots,G^{(T)}\}$ \cite{meyer24}.  We advocate that the
distance between graphs should be evaluated in the spectral domain, by comparing the
eigenvalues of the normalized Laplacian, $\cLt$, of the respective graphs $\Gt$. We define
the Laplacian spectral pseudo-metric between two graphs $G$ and $G^\prime$ by
\begin{equation}
  d(\cL,\cL^\prime) \eqdef \nrm{\blb (\cL) - \blb (\cL^\prime)}{2},
  \label{la-distance-Laplace}
\end{equation}
where $\blb (\cL)$ and $\blb (\cL^\prime)$ are the vectors of eigenvalues of $\cL$ and
$\cL^\prime$ respectively. This pseudo-distance captures at multiple scales the structural
and connectivity in the graphs (e.g., diameter, number of connected components,
clusterability, diffusion distance, \cite{Chung1997,Coifman06b,donnat18,wills20c}). Defining
a pseudo-distance in the spectral domain alleviates the difficulty of solving the node
correspondence problem, and in the case of the normalized Laplacian, it makes it possible to
compare graphs of different sizes. When the graphs are realizations of a stochastic block
model, the eigenvalues of $\cL$ associated with each community are better separated from the
bulk (see \cref{les_valeurs_propres} for a real-life instance of this fact) than the
corresponding eigenvalues of $\bL \eqdef \bD - \bA$ \cite{deng21}.

  \subsection{The spectrum of the barycentre graph}

In spite of the advantages of the pseudo-metric $d$ defined by \cref{la-distance-Laplace},
the computation of the solution to \cref{le-barycentre} leads to two technical obstacles.
The first challenge stems from the fact that $d$ is defined in the spectral domain, but the
optimization \cref{le-barycentre} takes place in $\cS$ \cite{garner24}. This leads to the
definition of a {\em realizable} sequence; we say that $\blb = [\lambda_1,\ldots,\lambda_n]$
is realizable if there exists $\bA \in \cS$ whose Laplacian, $\cL(\bA)$, satisfies
$\blb(\cL(\bA)) = \blb$. Hereinafter, we define $\cR$ to be the {\em set of realizable
  sequences}. We can formalize the optimisation problem associated with the estimation
of $\fmpr$ in \cref{le-barycentre},
\begin{equation}
  \spcL
  = \argmin{\blb  \in \cR }\sum_{t=1}^T ||\blb - \blb(\cLt)||_2^2.
  \label{def-mean-spectral}
\end{equation}
If we relax this minimization problem ($\blb \in \R^n$, instead of $\blb \in \cR$), then the
solution to \cref{def-mean-spectral}
is the sample mean $\sE{\blb} \eqdef T^{-1} \sum_{t=1}^T \blb(\cLt)$, which
has no guarantee to be realizable. Which brings us to the second
difficulty in using a spectral pseudo-distance. 

  \subsection{The normalized Laplacian of the barycentre graph}
  The knowledge of the eigenvalues $\spcL$ solution to \cref{def-mean-spectral} is
  insufficient to reconstruct the adjacency matrix of the barycentre graph $\fmpr$; we need an
  orthonormal basis of eigenvectors, $\bPsi= \begin{bmatrix} \bpsi_1 \cdots
    \bpsi_n \end{bmatrix}$ to compute 
  \begin{equation}
    \Lmu = \bPsi \diag{\sE{\blb}} {\mspace{-12mu}\bPsi^T}.
    \label{Laplacian_of_barycentre}
  \end{equation}
  In addition to the constraint,  $\bPsi \in \OG(n)$, where $\OG(n)$ is the orthogonal group, 
  we need to make sure that $\begin{bmatrix} \bpsi_1 & \cdots & \bpsi_n \end{bmatrix}$ are the
  eigenvectors of a valid normalized Laplacian, to wit 
  \begin{equation}
    \exists \bA \in \cS, \mspace{16mu} \bPsi \diag{\sE{\blb}} {\mspace{-12mu}\bPsi^T} = \id -
    \bD^{-1/2} \bA\bD^{-1/2}, \label{validLaplacian}
  \end{equation}
  where $\bD = \diag{\bA \one}\mspace{-16mu}$. When $\bPsi$ complies with
  \cref{validLaplacian}, we can take the adjacency matrix of the barycentre graph to be
  \begin{equation}
    \fmpr = \bD^{1/2}\big[\id - \bPsi \diag{\sE{\blb}}\mspace{-12mu} \bPsi^T\big]\bD^{1/2}.
    \label{from_Laplacian_to_barycentre}
  \end{equation}
  \subsection{The adjacency matrix of the barycentre graph: a structural constraint}

In general the barycentre $\fmpr$ given by \cref{from_Laplacian_to_barycentre}, for some
random choice of $\bPsi \in \OG(n)$, may have a very different topological structure than
the graphs $\big\{G^{(1)},\ldots, G^{(T)} \big\}$. We would like to enforce the fact that
$\fmpr$ should share the same topology and connectivity as $\E{\pr}$. This requirement is
well founded since $\E{\pr}$ captures the structural connectivity of the random graph
ensemble defined by $\pr$. For instance, if $\E{\pr}$ contains structures such as modular
communities, rich clubs, hubs, trees, etc. we expect these structures to be present in
$\fmpr$.

   This structural constraint can be formalized in the following way: we
  request that $\fmpr$ converges (as $n\rightarrow +\infty$) to $\E{\pr}$, with high
  probability (asymptotically almost-surely). Formally, we require
  \begin{equation}
    \lim_{n\rightarrow +\infty} \prob{\fmpr  =  \E{\pr}} = 1.
    \label{barycentre_close_to_mean}
  \end{equation}
  We note in passing, that the trivial choice $\fmpr = \mA$, where
  \begin{equation}
    \big[\sE{\pr}\big]_{ij} \eqdef T^{-1} \sum_{t=1}^T a^{(t)}_{ij},
    \label{sample_mean_P}
  \end{equation}
  does not satisfy the constraint given by \cref{def-mean-spectral}, since in general we
  have $\blb(\mA) \ne \sE{\blb}$ \cite{athreya22,chakrabarty20}.

  Because the topological structures present in a graph are independent of the normalization
  of $\bA$, we can replace $\bA$ with its symmetric normalized version $\hA$ (defined by
  \cref{Ahat}), or equivalently we can work with the normalized Laplacian of $\E{\pr}$,
  \begin{equation}
    \cLbar \eqdef \cL\big(\E{\pr}\big) = \id -\Big[\Dbar\Big]^{-1/2} \bP \Big[\Dbar\Big]^{-1/2},
  \end{equation}
  where
  \begin{equation}
    \bP \eqdef \E{\pr},
    \mspace{16mu} \text{and} \mspace{16mu}
    \Dbar \eqdef \E{\bD}.
    \label{edge_probability}
  \end{equation}
  Our definition of $\cLbar$ is in agreement with others who have used this matrix to derive
  bounds on the concentration of $\cL$ \cite{agterberg25,le18,lu13,oliveira09,rohe11}.  On
  this account, we replace the asymptotic structural constraint given by
  \cref{barycentre_close_to_mean} with
  \begin{equation}
    \lim_{n\rightarrow +\infty} \prob{\Lmu = \cLbar \:} = 1.
    \label{Laplacian_close_to_mean}
  \end{equation}
  \subsection{Formulation of the problem}
  We combine \cref{Laplacian_of_barycentre} and \cref{Laplacian_close_to_mean}, with
  the constraint that $\bpsi \in \OG(n)$ to get the following problem: given $\sE{\blb}$, find
  $\bPsi \in \OG(n)$ such that
  \begin{equation}
    \lim_{n\rightarrow +\infty} \prob{\bPsi \diag{\sE{\blb}} \mspace{-12mu}\bPsi^T =  \cLbar} = 1.
    \label{le_gros_probleme}
  \end{equation}
  In this work, we prove that it is possible to solve this problem using a ``customized''
  Soules basis $\bPsi$. We prove that when $\big(\cS,\pr\big)$ is the probability space
  associated with a balanced stochastic block model, then $\fmpr = \E{\pr}$ {
    asymptotically almost-surely}, where $\fmpr$ is defined by
  \cref{from_Laplacian_to_barycentre}.
  \subsection{State of the art}
  The requirement given by \cref{le_gros_probleme} can be approximately achieved if $\bPsi$ is
  an ``average on $\OG(n)$'' of the distribution of bases of eigenvectors associated with the
  respective Laplacian matrices $\cLt$ of the graphs in the sample. To this goal several
  authors have proposed to align the eigenvectors of the respective graph adjacency matrices
  \cite{ferrer05} or Laplacian matrices \cite{white07}. Rather than working with the
  eigenvalues of the normalized Laplacian, the authors in \cite{ginestet17} seek the
  barycentre graph on the cone of symmetric positive semidefinite matrices. They characterize
  the space of graph Laplacians, a subset of this cone, to develop statistical tests about the
  location of the \fr mean. Others \cite{ferguson23a} have proposed numerical methods to find
  the best stochastic block model (SBM) whose eigenvalues match the sample mean eigenvalues.
  \subsection{The stochastic block model\label{SBM-section}}
  Experiments performed on random graph ensembles, with adjustable parameters, provide a
  mechanism to quantify the performance of graph-valued algorithms \cite{penschuck22}.  To
  provide theoretical guarantees for the algorithms presented in this paper, we analyse the
  algorithms when the graphs are sampled from a random graph model: the stochastic block model
  \citeg{abbe18}.

  In the next subsections, we define the stochastic block model, justify the importance of
  this model, and give an example of a real-life dynamic graph that will be studied in
  detail in \cref{lespetitsfrancais}.

\subsubsection{Definition}
We define the general stochastic block model $\sbm{\bp}{q}{n}$. Let $\{B_k\}, 1 \le k \le M$
be a partition of the vertex set $[n]$ into $M$ blocks (or communities). We define the
vector $\bp = \begin{bmatrix} p_1 \cdots  p_M \end{bmatrix}$ to be the edge probabilities
within each block, and $q$ to be the edge probability between blocks. The entries $a_{ij} =
a_{ji}, i< j$ of the adjacency matrix $\bA$ are independent (up to symmetry) and are
distributed with Bernoulli distributions with parameter $p_m$ if $i$ and $j$ are in the same
block $B_m$, and parameter $q$ if $i$ and $j$ are in distinct blocks.

We often represent $\sbm{\bp}{q}{n}$ by the matrix of edge probabilities, or matrix of
connection probabilities, $\bP$ defined by \cref{edge_probability}. We sometimes consider a balanced
version of the model where all blocks have the same size, $\lon{B_m} = n/M$, (in that case
we assume without loss of generality that $n$ is a multiple of $M$), and all the edge
probabilities are equal, $p_1 = \cdots = p_M$. We denote this probability space by
$\sbm{p}{q}{n}$, since $p$ is a scalar and no longer a vector.

  \subsubsection{Motivation for using stochastic block models}

The stochastic block model represents the quintessential exemplar of a graph with community
structure. It has been used extensively in the study of complex real-life graphs
\cite{faskowitz18,thibeault24}. Stochastic block models have also been shown to provide
universal approximants (under various norms or distances) (e.g.,
\cite{airoldi13,ferguson23a,gerlach18,olhede14} and references therein), and can therefore
be used as building blocks to analyse more complex

  graphs. In fact, it was shown in \cite{young18} that the structural analysis of a large
  graph is equivalent to the approximation of that graph using stochastic block models,
  of decreasing complexity as the scale of the analysis becomes coarser.

Stochastic block models also provide a discrete version of step graphons
\cite{borgs20,dolevzal21,gao15}, which are dense in the space of graphons for the topology
induced by the cut-norm \cite{gao15}. Finally, stochastic block models are amenable to a
rigorous mathematical analysis, and are indeed at the cutting edge of rigorous probabilistic
analysis of random graphs \cite{abbe18}.

  \subsubsection{Example of a real-life graph that is well approximated with an SBM}

To emphasize the importance of the stochastic block model in the study of real-life
graphs, we present a dataset that is studied in detail in \cref{experiments}.  The
data is composed of a time-series of dynamic social-contact graphs collected in a French
primary school \cite{Stehle2011}. The dataset is considered a classic benchmark to study
real-life face-to-face contact graphs, and dynamical processes on such graphs (e.g.,
\cite{ferraz21,bovet22,cencetti25,contisciani22,djurdjevac25,failla24,gauvin14,gemmetto14,genois18,sattar23}).

Briefly, students carried RFID tags that recorded (every 20 seconds) face-to-face contacts
during two school days \cite{Stehle2011}. The primary school is composed of ten grades
(1-5); each grade is divided into two classes (A \& B); each student is a node of the
graph.  During the school day (8:30 AM -- 4:30 PM), events (morning and afternoon
recess: 10:30 -- 11:00 AM,
\begin{figure}[H]
  \setlength{\fboxsep}{-0.5pt}
  \centerline{
    \framebox{\parbox{.31\textwidth}{{\includegraphics[width=0.31\textwidth]
          {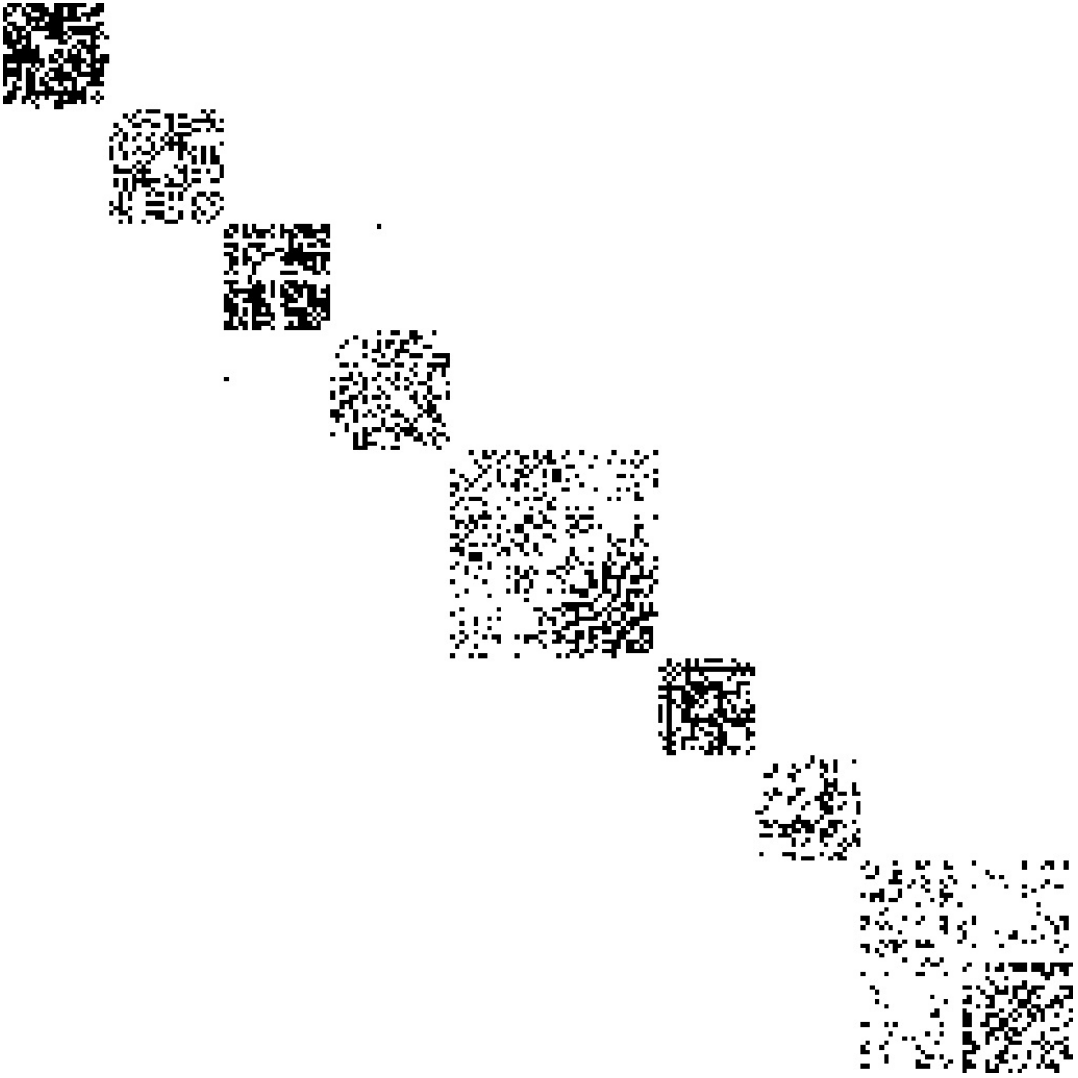}}}}
    \hspace*{0.4pc}
    \framebox{\parbox{.31\textwidth}{{\includegraphics[width=0.31\textwidth]
          {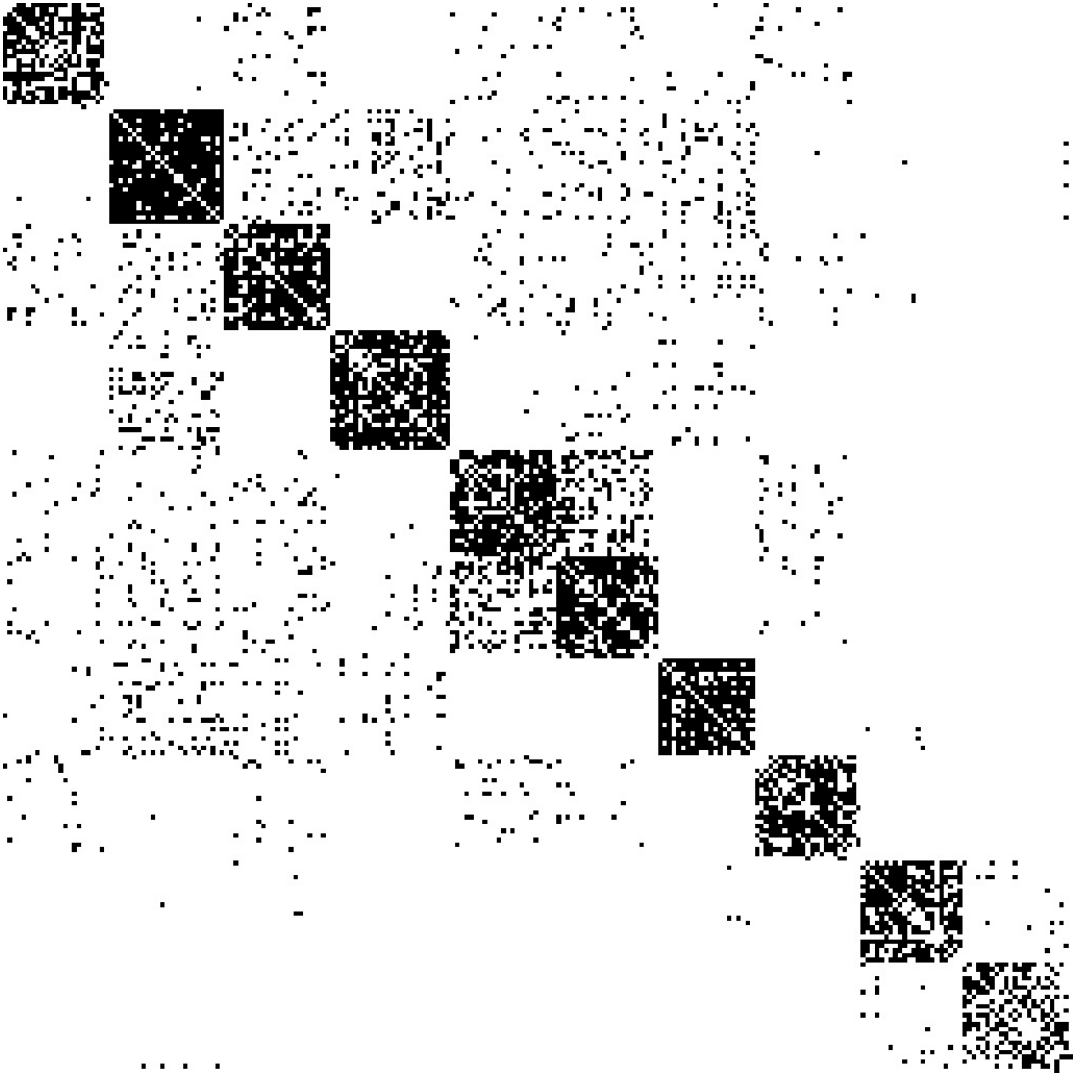}}}}
    \hspace*{0.4pc}
    \framebox{\parbox{.31\textwidth}{{\includegraphics[width=0.31\textwidth]
          {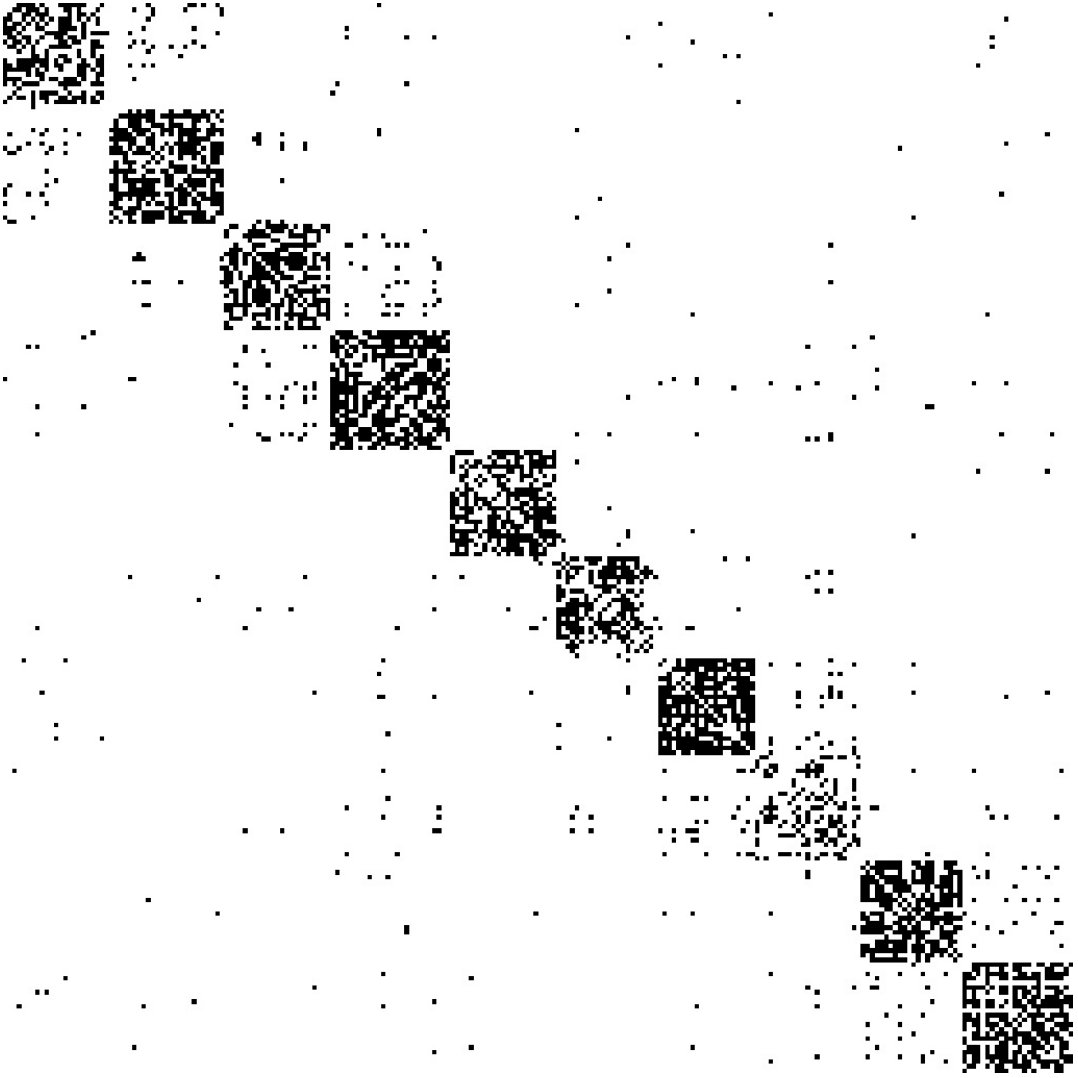}}}}
  }
  \centerline{
    \includegraphics[width=0.33\textwidth]{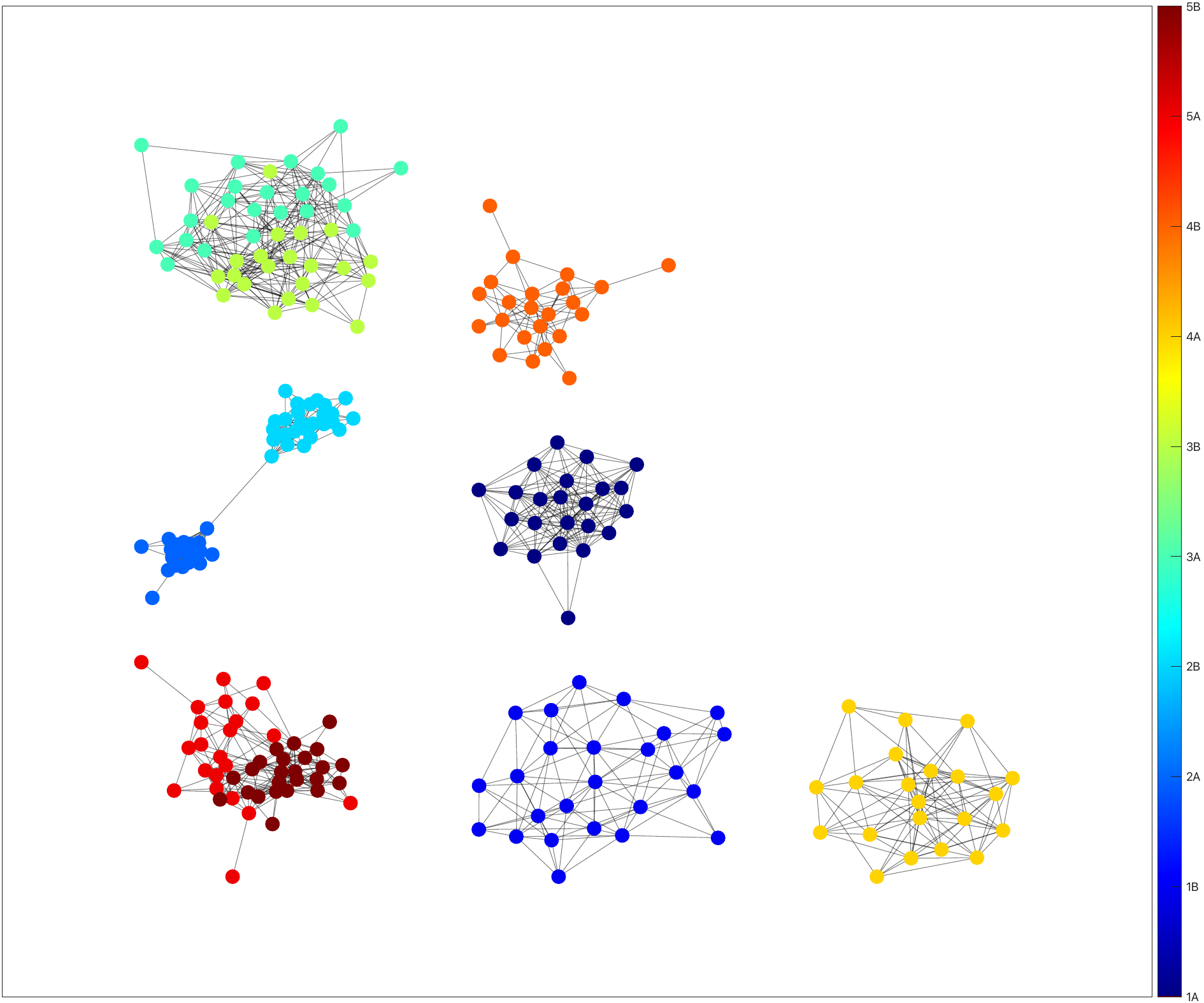}
    \includegraphics[width=0.33\textwidth]{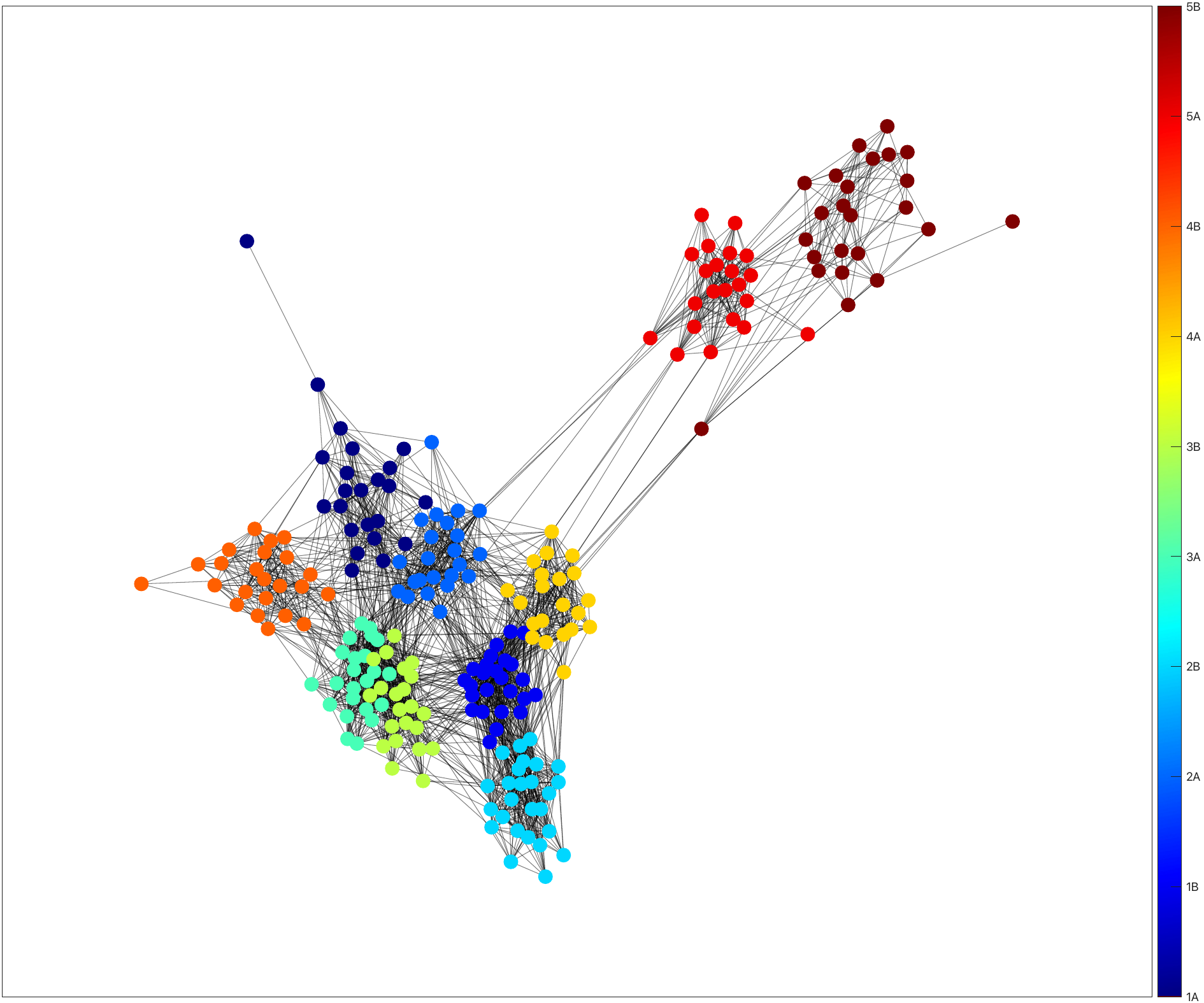}
    \includegraphics[width=0.33\textwidth]{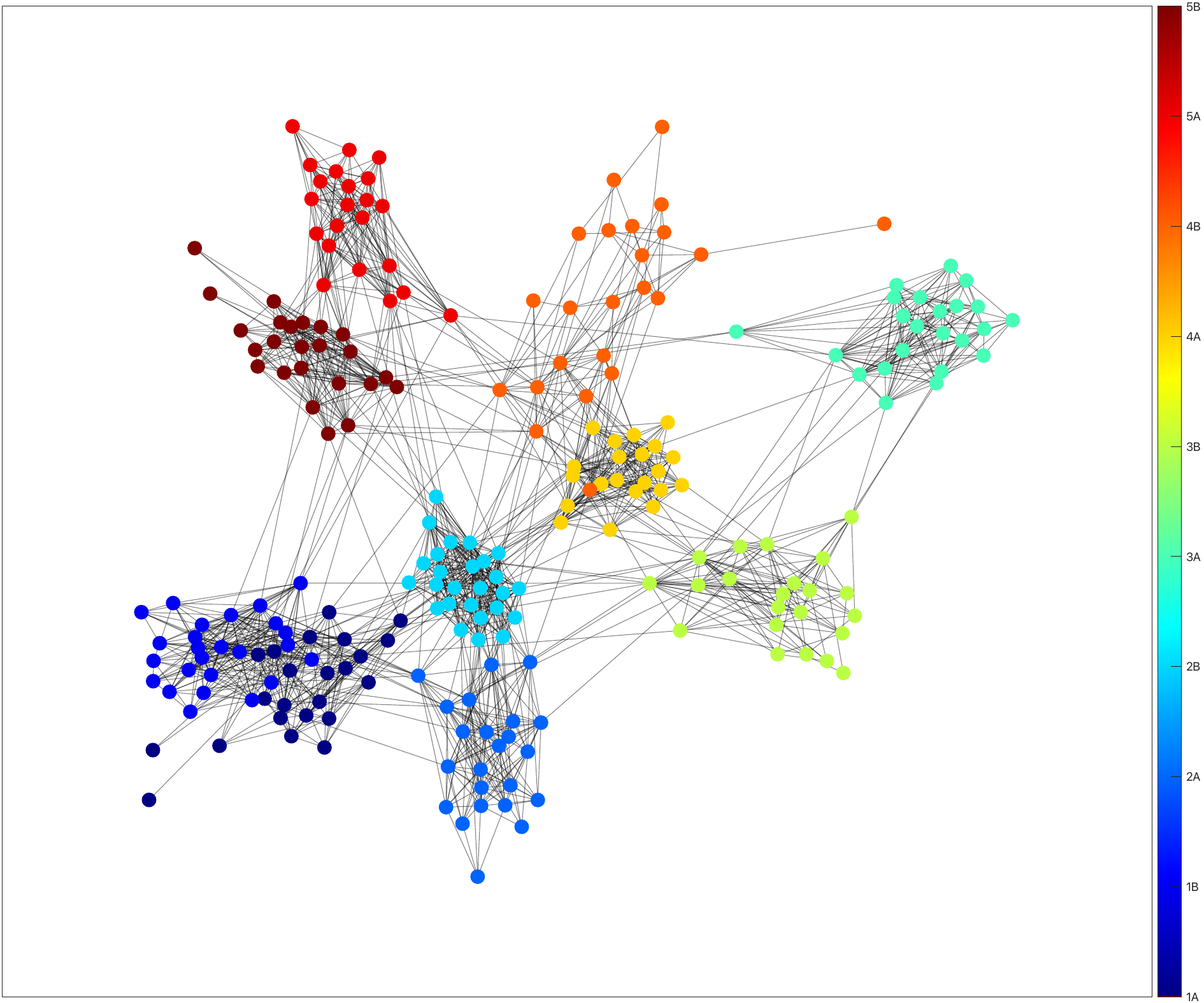}
  }
  \centerline{\hfill 9:00 AM \hfill \hfill 10:30 AM \hfill \hfill 12:00 PM \hfill}
  \centerline{\hfill beginning of the day\hfill \hfill morning recess\hfill\hfill end of morning
    period \hfill}
  \caption{Dynamic face-to-face contact graphs (top: adjacency matrix; bottom: graphical 
    representation). The nodes in the adjacency matrices are grouped from class 1A to class
    5B. Each node is a student; the color of the node encodes the class of the student (see
    colorbar). From left to right: beginning of the school day (9:00 AM); morning recess
    (10:30 AM), and end of morning period before lunch (12:00
    PM).\label{le_matin_a_l_ecole}}
\end{figure}
\noindent 3:30 -- 4:00 PM; lunch periods: 12:00 PM-- 1:00 PM, and 1:00 --
2:00 PM) punctuate the school day and trigger significant changes in connectivity and
topology of the contact graphs (see \cref{le_matin_a_l_ecole}). 

To facilitate the interpretation of the community structure, we aggregate the graphs over a
time window of $40$ minutes. The graphs at 9:00 AM (see \cref{le_matin_a_l_ecole}-left)
display the grouping of the two classes of the same grade (second, third, and fifth) caused
by a common grade-dependent activity. At 10:30 AM all students mix during recess (see
\cref{le_matin_a_l_ecole}-center). We note that students are preferably in contact with
other students from the same grade. A similar connectivity pattern is discernible before the
lunch break (see \cref{le_matin_a_l_ecole}-right).

\noindent The graphs displayed in \cref{le_matin_a_l_ecole} are examples of real-life graphs
that can be well approximated with stochastic block models. Blocks associated with denser
connectivity are clearly visible in the adjacency matrices of the graphs at 9:00 AM, 10:30
AM, and 12:00 PM (see top row of \cref{le_matin_a_l_ecole}). We take note of the merging of
the fifth and sixth blocks, and the ninth and tenth blocks, at 9:00 AM. This is also visible
in the graphical depiction of the graphs (see left of top and bottom rows of
\cref{le_matin_a_l_ecole}). These contact graphs exhibit intricate dynamics and complex
changes in connectivity and structural properties.

A more detailed analysis reveals that the distribution of eigenvalues of the
normalized graph Laplacian of these graphs appear similar to the theoretical distributions
of the eigenvalues of the normalized graph Laplacian of the stochastic block model (see
\cref{les_valeurs_propres}). Specifically, both distributions reveal the presence of a
bump-shaped, centered around 1, which is usually referred as the {\em bulk}, and which
contains most eigenvalues. The presence of the bulk in this dataset is created by the
stochastic nature of the graphs. A similar bulk is present in the empirical spectral
distribution of the stochastic block model
\cite{athreya22,avrachenkov22,chakrabarty20,oliveira09,zhang14}; see also
\cref{lambda_inthelimit} in \cref{informal_statements}. Starting at 0, there are ten
eigenvalues that are separated (this phenomenon is more visually
striking in the morning distribution) from the bulk in the morning and
afternoon distributions (see left of both distributions in \cref{les_valeurs_propres}),
which is a signature of the stochastic block model.
\begin{figure}[H]
  \noindent\centerline{
    \includegraphics[width=0.4\textwidth]{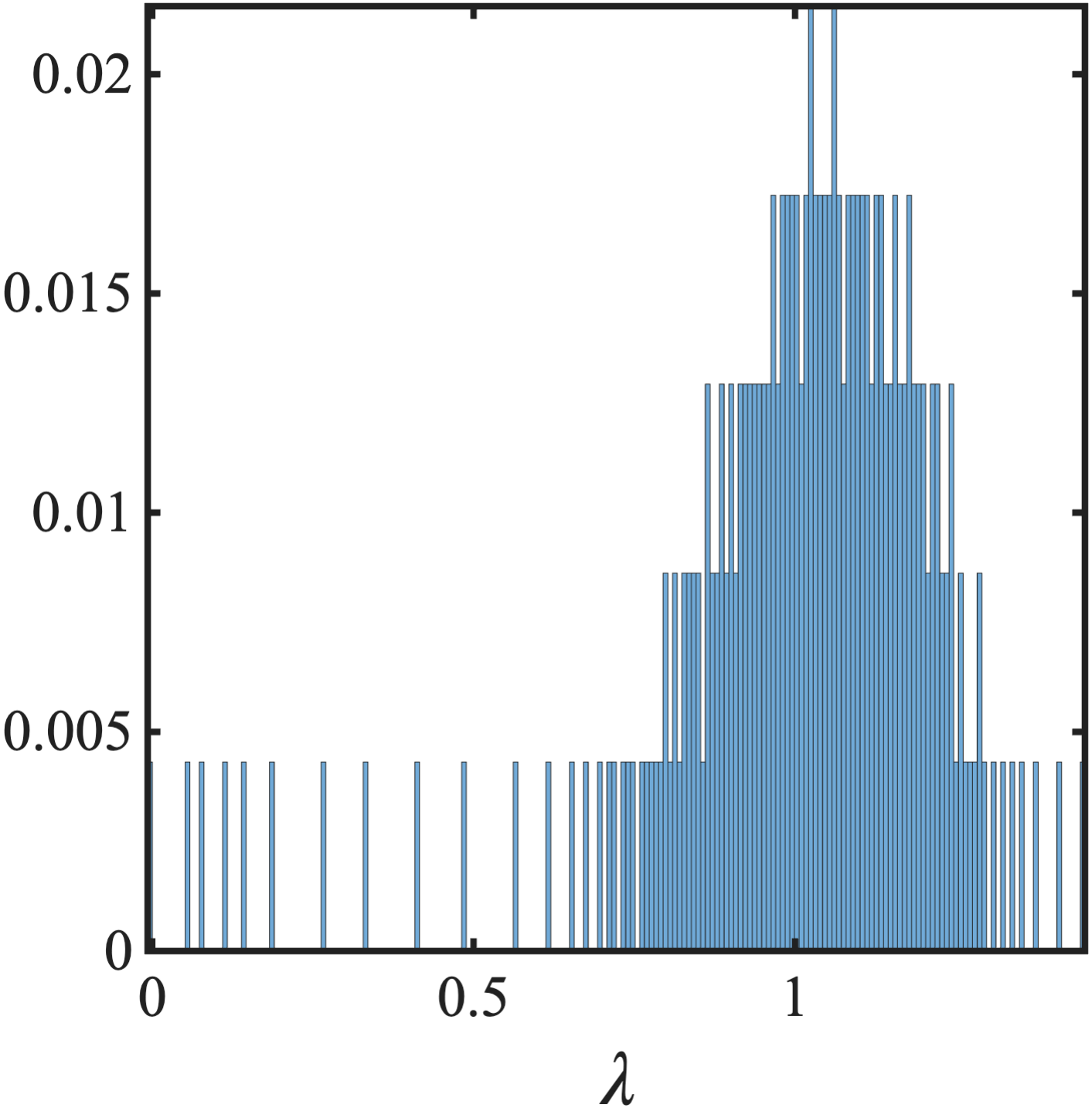}
    \hspace*{1pc}
    \includegraphics[width=0.4\textwidth]{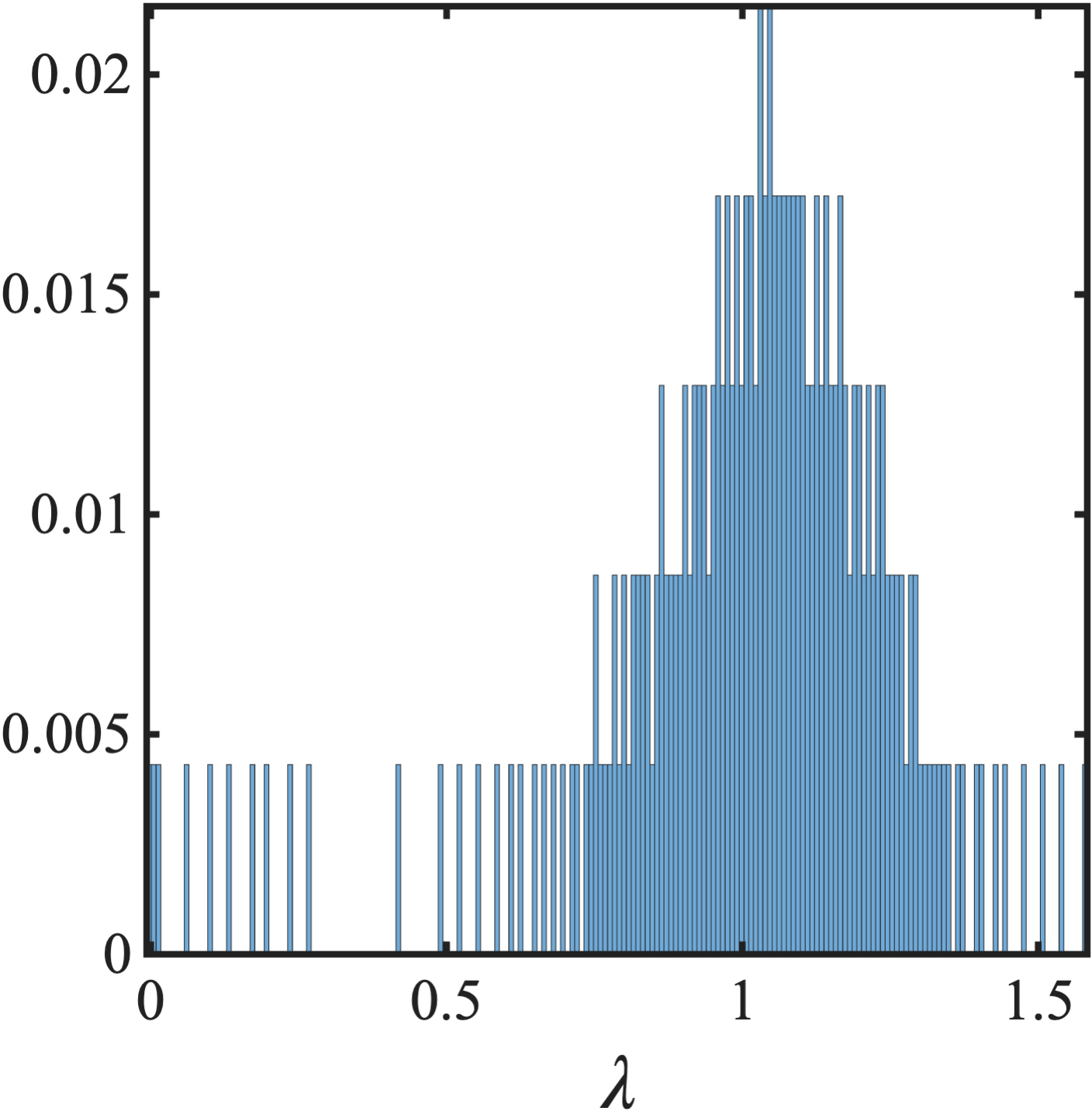}
  }
  \centerline{morning \hspace{8pc} afternoon}
  \caption{Distribution of the eigenvalues of $\cL$ in the morning (left) and afternoon
    (right). For both distributions the ten smallest eigenvalues are separated from the bulk.
    \label{les_valeurs_propres}}
\end{figure}
\subsection{Content of the paper: overview of the results}
The main contribution of this work is a fast algorithm to compute the barycentre of a set of
graphs with community structure. The barycentre is determined using a Laplacian spectral
pseudo-distance. We solve the problem defined by \cref{le_gros_probleme} by relaxing the
optimization problem \cref{def-mean-spectral} and taking

  { $\blb(\fmpr) = \sE{\blb}$.}

We then construct $\bPsi$ using the large library of Soules bases
\cite{eubanks10,soules83}, which were designed specifically to solve similar inverse
eigenvalue problems.
This work is significant because it opens the door to the design of
new spectral-based graph synthesis \cite{baldesi19,shine19} that have theoretical
guarantees. We publicly share our code to facilitate future work \cite{meyer25}.
\section{A tour of the key ideas
  \label{informal_statements}}
To help guide the intuition of the reader, we highlight the key ideas of our
approach. Except for Idea~2, we omit the proofs; \cref{BSB-section}  comprises the rigorous
execution of these ideas. 

We provide a description of our approach in the most restricted context wherein we can
derive proofs for the theorems: the graphs in the sample are random realizations of a
balanced stochastic block model, $\sbm{p}{q}{n}$. Our experiments (see \cref{experiments})
demonstrate that in practice our approach works beyond the controlled environment of
balanced stochastic block models; these experiments suggest that our theoretical analysis
could be extended to a larger class of community graphs.\\

  ~\\
        {\noindent \bfseries Idea~1.} 
        To solve \cref{le_gros_probleme}, we first derive the expression of $\cLbar$ in the case
        where $\pr$ is a general $\sbm{\bp}{q}{n}$. A simple calculation yields
        \begin{equation}
          \cLbar_{ij} =  
          \begin{cases}
            \displaystyle - \frac{p_m}{\dbar_m} & \text{if}\; (i,j) \in B_m\times B_m, \\
            1  & \text{if}\;  i = j,\\
            \displaystyle - \frac{q}{\sqrt{\dbar_m\dbar_\mpr}} & \text{if}\; (i,j) \in B_m\times B_\mpr, m \ne \mpr,
          \end{cases}
          \label{expected_normalized_Laplacian}
        \end{equation}
        where $\dbar_m \eqdef \E{d_{ii}} = \lon{B_m}\big\{p_m + (n/\lon{B_m} -1)q\big\}$ for all nodes $i \in
        B_m$.
        \vspace*{2em}\\
                {\noindent \bfseries Idea~2.} 
                In the case where $\pr$ is a general $\sbm{\bp}{q}{n}$, the first $M$ eigenvalues of $\cL$, $0 =
                \lambda_1 \le \cdots \le \lambda_M$ converge to deterministic values as the size of the
                graph goes to infinity
                \cite{athreya22,avrachenkov22,chakrabarty20,fan22,lowe24,lu13,rohe11,tang18}. We propose 
                simplifying \cref{le_gros_probleme} by substituting $\sE{\lambda_k}$ for the large graph
                size deterministic limits.

                Unfortunately, closed form expressions of these limits can only be derived for balanced
                $\sbm{p}{q}{n}$ composed of $M$ blocks \cite{fan22,lowe24,rohe11}. For this reason, hereafter we assume
                that $\pr = \sbm{p}{q}{n}$. In \cite{lowe24}, the authors provide the following estimates of
                the eigenvalues $\blb(\cL)$ of the normalized Laplacian.
                \begin{lemma}[Proposition 4.3 of \cite{lowe24}]
                  \label{lemmaLowe24}
                  Let $\bA$ be a random realization of $\sbm{p}{q}{n}$, with normalized Laplacian matrix
                  $\cL$. Then the eigenvalues of $\cL$ are given by
                  \begin{equation}
                    \lambda_m = l_m  + \O{\sqrt{\frac{\log n}{n}}},
                  \end{equation}
                  \noindent asymptotically almost-surely, where
                  \begin{equation}
                    l_m  =
                    \begin{cases}
                      0  & \text{if} \mspace{8mu} m=1,\\
                      \displaystyle \frac{Mq}{p + (M-1)q} & \text{if} \mspace{8mu} 2  \le m \le M,\\   
                      1 & \text{if} \mspace{8mu} M+1 \le m \le n.
                    \end{cases}
                    \label{lambda_inthelimit}
                  \end{equation}
                \end{lemma}
                This result informs the role of the sample eigenvalues $\sE{\lambda_m}$. We expect that only
                the lowest $M$ eigenvalues of $\cL$ contribute to the recovery of the geometry and the edge
                density of $\sbm{p}{q}{n}$. Conversely, we anticipate that the largest $\sE{\lambda_m}$,
                $M+1 \le m \le n$, which are created by the stochastic nature of the model and appear in the
                bulk, \cite{avrachenkov22,chakrabarty20,chaudhuri12,lee14,zhang14}, do not serve a useful
                purpose and only contaminate the reconstruction of $\Lmu$ in \cref{Laplacian_of_barycentre}.

                To alleviate this issue, we propose to replace the $n-M$ largest sample eigenvalues
                $\sE{\lambda_m}$ with their limits $l_m$, and define a ``regularized'' set of eigenvalues as
                follows,
                \begin{equation}
                  \hlm_m =
                  \begin{cases}
                    \sE{\lambda_m} & \text{if} \; 1 \le m \le M,\\
                    l_m & \text{if} \; M+1 \le m \le n.\\
                  \end{cases}
                  \label{regularized}
                \end{equation}
                In practice, one needs to estimate $M$, the number of eigenvalues outside the
                bulk. Fortunately, many estimators are available \citeg{saldana17,fan20,yan18}. We use the
                regularized eigenvalues $\{\hlm_m\}$ to devise the following estimator of $\Lmu$,
                \begin{equation}
                  \hLM \eqdef \sum_{k=1}^n \hlm_k \bpsi_k\bpsi_k^T =
                  \sum_{k=1}^M \sE{\lambda_k} \bpsi_k\bpsi_k^T +
                  \mspace{-8mu} \sum_{k=M+1}^n \mspace{-4mu} l_k \bpsi_k\bpsi_k^T.
                  \label{truncated_Laplacian}
                \end{equation}
                The following lemma provides sufficient conditions for $\hLM$ to converge (in the limit of
                large graph size, $n \rightarrow +\infty$) toward $\cLbar$ with high probability.
                \begin{lemma}
                  \label{leprobleme-whp}
                  If the eigenvectors
                  $\bPsi= \begin{bmatrix} \bpsi_1 & \cdots & \bpsi_n \end{bmatrix}$ satisfy
                  \begin{equation}
                    \left\{
                    \begin{aligned}
                      & \bpsi_1 = n^{-1/2} \one, \\
                      & \sum_{k=1}^M \bpsi_k\bpsi_k^T (i,j) = 
                      \begin{cases}
                        M/n & \text{if} \mspace{8mu}  \exists \; m \in [M], \; (i,j) \in B_m\times B_m, \\
                        0 & \text{otherwise,}
                      \end{cases}\\
                      &\sum_{k=1}^n \bpsi_k\bpsi_k^T = \id,    
                    \end{aligned}
                    \right. \label{le_gros_probleme3}
                  \end{equation}  
                  then we have
                  \begin{equation}
                    \nrm{\hLM - \cLbar}{F} = \o{1},
                  \end{equation}
                  with probability converging to 1 as the graph size $n \rightarrow +\infty$. In other words,
                  if we replace the sample eigenvalues $\sE{\lambda_m}$ with the ``regularized'' eigenvalues $\hlm_m$ (defined
                  by \cref{regularized}), then the matrix $\fmpr$ corresponding to $\hLM$ satisfies
                  \cref{le_gros_probleme}.
                \end{lemma}
                \begin{proof}
                  We first specialize \cref{expected_normalized_Laplacian} for $\sbm{p}{q}{n}$, to get
                  \begin{equation}
                    \cLbar = \id - \displaystyle \frac{M}{n(P + (M-1)q)} \bP,
                  \end{equation}
                  where $\bP$ is defined in \cref{edge_probability}. We then substitute in
                  \cref{truncated_Laplacian} $\sE{\lambda_m}, \; 1 \le m \le M$ for the large graph size
                  estimate given by \cref{lambda_inthelimit}. We obtain with probability converging to 1 as
                  the graph size $n \rightarrow +\infty$,
                  \begin{equation}
                    \hLM = \sum_{k=1}^n l_k \bpsi_k\bpsi_k^T  + \eps_n,
                  \end{equation}
                  where $\eps_n = \O{\sqrt{\displaystyle \frac{\log n}{n}}} \bigg\{\sum_{k=1}^M \bpsi_k\bpsi_k^T\bigg\}$.\\

                  \noindent Now,
                  \begin{equation}
                    \begin{aligned}
                      \sum_{k=1}^n l_k \bpsi_k\bpsi_k^T
                      & =  \frac{Mq}{p + (M-1)q} \sum_{k=2}^M \bpsi_k\bpsi_k^T + \sum_{k=M+1}^n  \bpsi_k\bpsi_k^T\\
                      & =  \frac{Mq}{p + (M-1)q} \sum_{k=1}^M \bpsi_k\bpsi_k^T  - \frac{Mq}{p + (M-1)q} \bpsi_1\bpsi_1^T 
                      - \sum_{k=1}^M  \bpsi_k\bpsi_k^T + \sum_{k=1}^n  \bpsi_k\bpsi_k^T.
                    \end{aligned}
                  \end{equation}
                  Whence
                  \begin{equation}
                    \hLM =  \sum_{k=1}^n \bpsi_k \bpsi_k^T
                    - \bigg\{
                    \frac{p-q}{p + (M-1)q} \big( \sum_{k=1}^M \bpsi_k \bpsi_k^T\big)
                    + \frac{Mq}{p + (M-1)q} \bpsi_1\bpsi_1^T
                    \bigg\}
                    + \eps_n,
                  \end{equation}
                  with high probability. If the eigenvectors $\bPsi= \begin{bmatrix} \bpsi_1  \cdots 
                    \bpsi_n \end{bmatrix}$ satisfy \cref{le_gros_probleme3} then a simple calculation shows
                  that
                  \begin{equation}
                    \hLM - \cLbar = \eps_n,
                  \end{equation}
                  with probability converging to 1 as the graph size $n \rightarrow +\infty$. Finally, since
                  $\bPsi$ is an orthonormal basis, we get
                  \begin{equation}
                    \nrm{\eps_n}{F} = \O{\sqrt{\displaystyle \frac{\log n}{n}}} \nrm{\sum_{k=1}^M\bpsi_k
                      \bpsi_k^T}{F} = \O{\sqrt{\displaystyle \frac{\log n}{n}}} \sqrt{M}.
                  \end{equation}
                  We conclude that $\nrm{\hLM - \cLbar}{F} = \o{1}$, with probability converging to 1 as the
                  graph size $n \rightarrow +\infty$.
                \end{proof}

~\\
{\noindent \bfseries Idea~3.}
Given a sample mean estimate $\mA$ of $\E{\pr}$, we design an algorithm that explores
the library of Soules bases (which is organized as a binary tree \cite{elsner98}), and
returns an orthonormal Soules basis $\bPsi = \begin{bmatrix} \bpsi_1\cdots
  \bpsi_n \end{bmatrix}$, that satisfies \cref{le_gros_probleme3} 

We briefly describe the ideas behind the construction of the sequence of $\bpsi_k$ in the
library of Soules basis. Because the support of $\sum_{k=1}^M \bpsi_k\bpsi_k^T$ is formed by
the $M$ blocks of the $\sbm{p}{q}{n}$ (see the second condition in
\cref{le_gros_probleme3}), we design $\bpsi_2, \bpsi_3, \ldots,\bpsi_M$ so that they are
constant on each block $B_m$; and the zero-crossing of $\bpsi_k$ is aligned with the jumps
between the blocks in $\mA$.
\begin{figure}[H]
  \centerline{
    \includegraphics[width=0.3\textwidth]{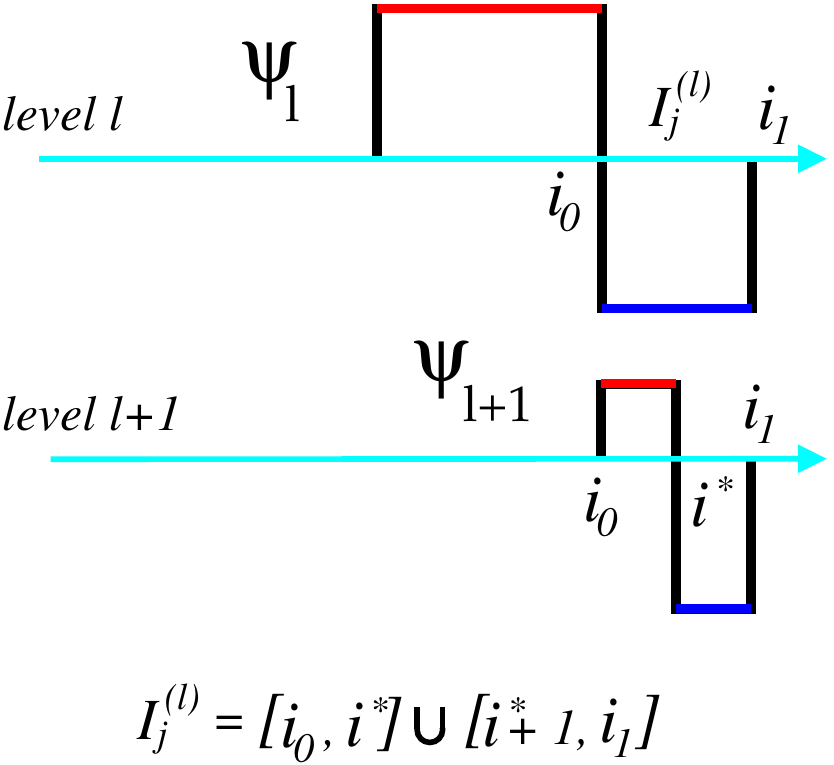}
    \includegraphics[width=0.2\textwidth]{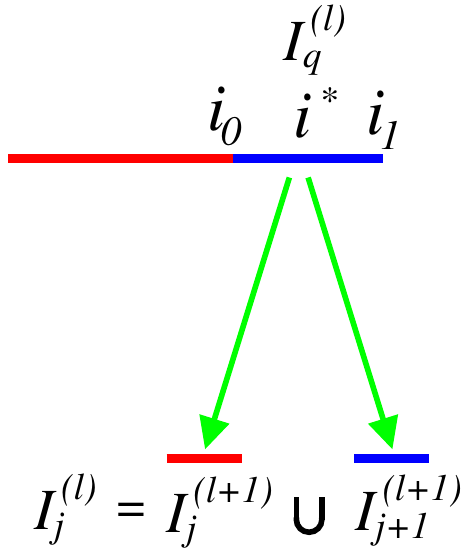}
  }
  \caption{Left: $\bpsi_{l+1}$ is created by splitting the block of indices
    $I^{(l)}_q=[i_0,i_1]$ at level $l$ into two sub-blocks, $[i_0,\ist] \cup [\ist+1,i_1]$ at level
    $l+1$.  Right: a node in the Soules binary tree is triggered by the splitting of 
    $[i_0,i_1]= [i_0,\ist] \cup [\ist+1,i_1]$.
    \label{one-iteration}}
\end{figure}
\noindent The construction of the Soules vectors starts at the coarsest scale with $\bpsi_1=
\sqrt{n}\one$.  The next Soules vector, $\bpsi_2$, is designed to detect the largest
gradient between any pair of blocks $B_m \times B_m$ and $B_\mpr \times B_\mpr$. In terms of
$\mA$, the zero-crossing of $\bpsi_2$ is carefully aligned with the boundaries between two
blocks of $\mA$ associated with the largest jump in the edge probability. Whence we can
choose $\bpsi_2$ to maximize $\lvert \ipr{\bpsi_2 \bpsi_2^T}{\mA}\rvert^2$. The construction
of the remaining $\bpsi_k$ proceeds iteratively by detecting all the boundaries between the
blocks $B_m \times B_m$. An added benefit of working with Soules bases is that the condition
$\sum_{k=1}^n \bpsi_k\bpsi_k^T = \id$ comes for free \cite{elsner98}.
\subsection{Organization of the paper}

   For the sake of completeness, we review in \cref{les_bases_soules} the
  concept of Soules bases and their key theoretical properties. The reader who is already
  familiar with Soules bases can skip to \cref{BSB-section} wherein we describe the
  algorithm that computes the Soules basis solution to \cref{le_gros_probleme3}. The
  theoretical properties of \cref{BSB} is presented in
  \cref{la_section_des_details}. Finally, we report results of experiments on synthetic and
  real-life datasets in \cref{experiments}.  In \cref{la-discussion}, we discuss the
  implications of our work. The proofs of some technical lemmata are left aside in
  \cref{les-preuves}.
\section{Soules Bases: definition and properties
  \label{les_bases_soules}
}
Soules bases \cite{elsner98,soules83} were invented to provide a solution to the following
symmetric nonnegative inverse eigenvalue problem: find $\bPsi \in \OG(n)$ such that
$\bPsi\diag{\lambda_1,\ldots,\lambda_n}{\mspace{-12mu}\bPsi^T} \in \cS$. Soules bases
provide a large family of solutions to this problem.
\subsection{Definition of Soules Bases}
A Soules bases is an orthogonal matrix that is constructed iteratively by applying a
product of Givens rotations to a fixed vector $\bpsi_1$ with nonnegative entries.  The
construction starts at the coarsest level ($l=1$) with a normalized vector $\bpsi_1$ with
nonnegative entries, whose support is the interval $I^{(1)} = [n]$. Hereunder, our analysis
assumes that we always choose $\bpsi_1 \eqdef \scalebox{0.8}{$\small n^{-1/2}$} \one$.

At any given level $l$, the set $[n]$ is partitioned into $l$ ordered intervals $I^{(l)}_q,
1 \le q \le l$.  When progressing from level $l$ to $l+1$, one chooses an interval,
$I^{(l)}_q=[i_0,i_1)$, and one chooses an index $\ist \in [i_0,i_1]$ and defines
  $I_q^{(l+1)} \eqdef [i_0,\ist]$, and $I_{q+1}^{(l+1)} \eqdef [\ist+1,i_1]$ (see
  \cref{one-iteration}-right).  The split of $I^{(l)}_q$ into $I^{(l+1)}_q$ and
  $I^{(l+1)}_{q+1}$ triggers the construction of the Soules vector $\bpsi_{l+1}$ (see
  \cref{one-iteration}-left), defined by \begingroup \addtolength{\jot}{.25\baselineskip}
  \begin{equation}
    \bpsi_{l+1} (i) \eqdef
    \frac{1}{\norm{\bpsi(i_0:i_1)}}
    \left\{
    \begin{aligned}
      \mspace{14mu} \displaystyle \frac{\norm{\bpsi_1(\ist+1:i_1)}}{\norm{\bpsi_1(i_0:\ist)}}
      \bpsi_1 (i) \mspace{16mu} & \text{if} \mspace{16mu} i_0 \le i \le \ist,\\
      - \displaystyle \frac{\norm{\bpsi_1(i_0:\ist)}}{\norm{\bpsi_1(\ist+1:i_1)}} \bpsi_1
      (i) \mspace{16mu} & \text{if} \mspace{16mu} \ist+1 \le i \le i_1,\\
      0\mspace{128mu}  & \text{otherwise,}
    \end{aligned}
    \right.
    \label{from_l_to_plusone}
  \end{equation}
  \endgroup
  where the vectors $\bpsi_1(i_0:i_1), \bpsi_1(i_0:\ist)$, and $\bpsi_1(\ist+1:i_1)$ are
  $n$-dimensional vectors whose nonzero entries are extracted from $\bpsi_1$ at the corresponding indices,
  \begin{equation}
    \scalebox{0.9}{$
      \begin{aligned}
        \bpsi_1(i_0:i_1) & = [0 \; \cdots \; 0 \mspace{12mu} \bpsi_1(i_0) \; \cdots \;
          \bpsi_1 (\ist) \; \bpsi_1(\ist+1) \; \cdots \; \bpsi_1 (i_1)\; 0  \; \cdots 0]^T,\\
        \bpsi_1(i_0:\ist) & = [0 \; \cdots \; 0 \mspace{12mu} \bpsi_1(i_0) \; \cdots \;
          \bpsi_1 (\ist) \mspace{8mu} 0\cdots\cdots\cdots\cdots\cdots\cdots\cdots
          \cdots \mspace{2mu} 0]^T,\\
        \bpsi_1(\ist+1:i_1) & = [0 \mspace{8mu} \cdots\cdots\cdots\cdots\cdots\cdots\cdots  0  \mspace{8mu} 
          \bpsi_1(\ist+1) \; \cdots \; \bpsi_1 (i_1)\; 0 \; \cdots 0]^T.
      \end{aligned}
      $}
  \end{equation}
 The iterative subdivision process can be described using a binary tree (see
  \cref{the-soules-tree}-right) where a new vector is created at each node that has two
  children. We observe that $\bpsi_k$ and $\bpsi_{k^\prime}$, $k \ne k^\prime$, are
  either nested, or they do not overlap; whence $\ipr{\bpsi_k}{\bpsi_{k^\prime}} = 0$, and
  $\begin{bmatrix} \bpsi_1 \cdots \bpsi_n \end{bmatrix}$ is an orthonormal matrix
  \cite{elsner98}; see \cite{thiele96} for a similar construction of Walsh-Hadamard
  packets and \cite{coifman01} for an analogous construction of complex-valued packets.
  \subsection{Properties of Soules Bases}
  Using \cref{from_l_to_plusone}, we derive the following lemma with a proof by induction.
  \begin{lemma}[See \cite{elsner98}]
    Let $\begin{bmatrix} \bpsi_1 \cdots \bpsi_n\end{bmatrix}$ be a Soules basis
      constructed according to \cref{from_l_to_plusone}. Then,
      \begin{equation}
        \forall m=1,\ldots,n, \mspace{16mu}
        \sum_{k=1}^m \bpsi_k \bpsi_k^T \ge 0,
        \mspace{8mu} \text{and} \mspace{16mu}
        \sum_{k=1}^n \bpsi_k \bpsi_k^T = \id.
        \label{definition-de-Em}
      \end{equation}
  \end{lemma}
Finally, we have the fundamental property of Soules bases.
  \begin{lemma}[See \cite{elsner98}]
    \label{nonnegative}
    Let $\bPsi$ be a Soules basis constructed according to \cref{from_l_to_plusone}.  Let
    $\bLambda=\diag{\lambda_1,\ldots,\lambda_n}\mspace{-12mu}$, where $\lambda_1 \ge \lambda_2 \ge \cdots
    \ge \lambda_n$. Then, the off-diagonal entries of $\bPsi \bLambda \bPsi^T$ are
    non-negative. In addition, if $\lambda_n \ge 0$, then $\bPsi \bLambda \bPsi^T \ge 0$.
  \end{lemma}
  \noindent We note that there has been some recent interest in Soules bases to solve various inverse
  eigenvalue problems \cite{devriendt19,redko20}.
  \begin{remark}
    The result in \cref{nonnegative} relies on the fact that the sequence of eigenvalues
    is decreasing. On the other hand, the eigenvalues of $\cL$ are by nature ranked in
    ascending order (the index $k$ of eigenvalue $\lambda_k$ of $\cL$ encodes the frequency of
    the corresponding eigenvector). Given an ascending sequence of eigenvalues of $\cL$, $0 =
    \lambda_1 < \lambda_2 \le \cdots \le \lambda_n$, we would like to apply
    \cref{nonnegative} to reconstruct a Laplacian matrix using a Soules basis. Since the
    off-diagonal entries of a normalized Laplacian $\cL$ are nonpositive, we need to work with
    $-\blb$. Then, $0 = \lambda_1 > - \lambda_2 \ge \cdots \ge - \lambda_n$, and we can use
    \cref{nonnegative} to construct $\cL^\ast$ such that
    \begin{equation}
      \cL^\ast = \bPsi\diag{\lambda_1,\ldots,\lambda_n}\mspace{-14mu}\bPsi^T, \mspace{16mu}
      \text{where} \mspace{8mu}
      \mathscr{L}^\ast_{ij} \le 0 \mspace{8mu}\text{if}\mspace{8mu}
      i \ne j.
      \label{synthesis}
    \end{equation}
    Since we choose, $\bpsi_1 = n^{-1/2}\one$, we have $\cL^\ast \one = \zr$, and therefore
    $\mathscr{L}^\ast_{ii} \ge 0$. While the signs of the entries of $\cL^\ast$ match
    those of a normalized Laplacian, there is no guarantee that $\cL^\ast$ be a valid
    normalized Laplacian (but see a definite answer in the case of the combinatorial
    Laplacian, $\bL = \bD -\bA$ in \cite{devriendt19}). 
  \end{remark}
  \begin{figure}[H]
    \centerline{
      \includegraphics[width=0.67\textwidth]{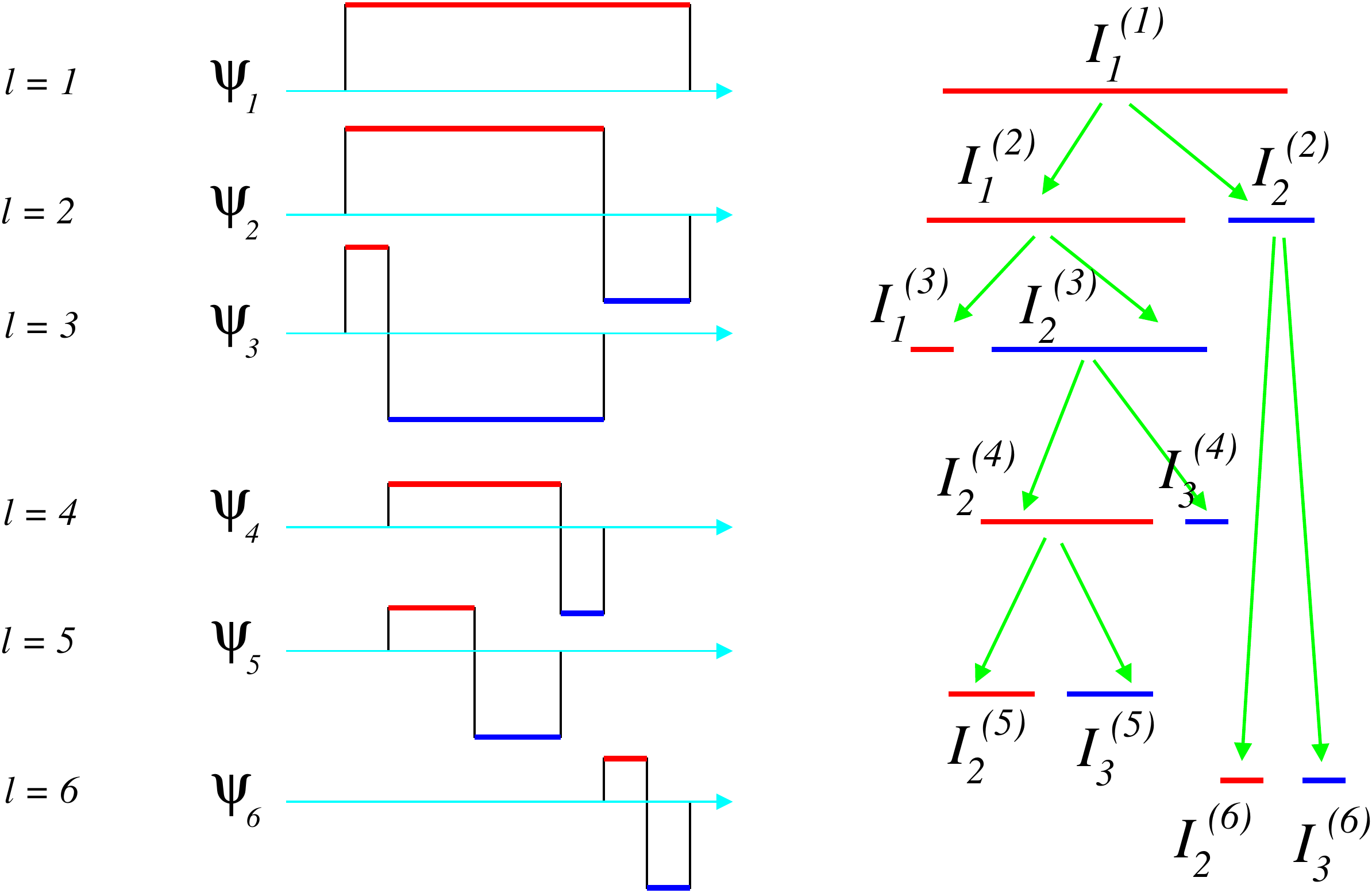}}
    \caption{Left: starting from level $l=1$, one Soules vector $\bpsi_l$ is constructed at
      each level $l\ge 2$ by selecting and then splitting an interval $I^{(l)}_j$ over which
      $\forall 1\le k \le  l, \;\rstr{\bpsi_k}{I^{(l)}_j}$ is constant. Right: each
      Soules basis is associated with a binary tree. The leaves of the tree are intervals
      that are not split.
      \label{the-soules-tree}}
  \end{figure}
    \section{A Soules Basis solution to \cref{le_gros_probleme3}
    \label{BSB-section}
  }
  \subsection{A greedy exploration of the Soules library}
  Given $T$ independent realizations of $\sbm{p}{q}{n}$, represented by their adjacency
  matrices $\bA^{(1)},\ldots,\bA^{(T)}$, we describe in this section a greedy algorithm that
  constructs a Soules basis $\begin{bmatrix}\bpsi_1  \cdots  \bpsi_n \end{bmatrix}$ that
  solves \cref{le_gros_probleme3}.
  
  As explained in \cref{informal_statements} (Idea~3), the construction of the Soules vectors
  proceeds using a multiscale approach. One starts at the coarsest scale with $\bpsi_1
  = {n}^{-1/2} \one$.  The next vector, $\bpsi_2$, is chosen so that it detects the two
  communities with the largest difference in edge probability. In terms of the sample
  adjacency matrix, the zero-crossing of $\bpsi_2$ is aligned with the boundaries between two
  blocks of $\sE{\pr}$ associated with the largest jump in the edge probability.

  Wherefore we choose $\bpsi_2$ to maximize $\lvert \ipr{\bpsi_2 \bpsi_2^T}{\mA}\rvert^2$. The
  next Soules vector, $\bpsi_3$, has its support inside either one of the two sets $\{\bpsi_2
  \ge 0\}$ or $\{\bpsi_2 \le 0\}$. We can therefore detect the second largest jump in the edge
  probability by maximizing the magnitude of the inner product between $\bpsi_3\bpsi_3^T$ and
  the reconstruction error, $\big[\mA - \ipr{\mA}{\bpsi_2\bpsi_2^T}\bpsi_2\bpsi_2^T \big]$,
  \begin{equation}
    \begin{aligned}
      \bpsi_3 & = \argmax{\bpsi_3 \text{defined by} \cref{from_l_to_plusone}}
      \big \lvert
      \ipr{\bpsi_3 \bpsi_3^T}
          {\big[\mA - \ipr{\mA}{\bpsi_2\bpsi_2^T}\bpsi_2\bpsi_2^T \big]}
          \big \rvert^2, \\
          & = \argmax{\bpsi_3 \text{defined by} \cref{from_l_to_plusone}}
          \big \lvert
          \ipr{\bpsi_3 \bpsi_3^T}{\mA}
          \big \rvert^2,
    \end{aligned}
  \end{equation}
  since $\bpsi_3^T\bpsi_2 = 0$.  At level $l$, the algorithm selects the new Soules vector
  $\bpsi_{l+1}$ that maximizes the magnitude of the inner product between
  $\bpsi_{l+1}\bpsi_{l+1}^T$ and $\mA$,
  \begin{equation}
    \bpsi_{l+1} = \argmax{\bpsi_{l+1} \text{defined by} \cref{from_l_to_plusone}}
    \bigg \lvert \ipr{\bpsi_{l+1}\bpsi_{l+1}^T}{\mA} \bigg
    \rvert^2.
  \end{equation}
  \noindent \cref{BSB} provides the pseudocode for the implementation of this algorithm.
  \begin{algorithm}[H]
    \caption{Top-down exploration of the Soules binary tree.\label{BSB}}
    \begin{algorithmic}[1]
      \Procedure{BestSoulesBasis}{$\mA$,$\bPsi$}
      \Comment{Input: $\mA$; Output: $\bPsi$ the Soules
        matrix}
      \State $\bpsi_1 \gets {n}^{-1/2} \one$
      \ForAll{levels $l\in\{1,\ldots, n-1\}$}  
      \Statex \Comment{For each block $I^{(l)}_q = [i_0,i_1]$ which was not split at level $l$ we split it
        using an index $\ist \in [i_0,i_1]$ and construct the eigenvector $\bpsi_{l+1}$
        associated with the split. We compute  $\lon{\ipr{\bpsi_{l+1}\bpsi_{l+1}^T}{\mA}}^2$}
      \State $\ist \gets 1$ \Comment{index of the tentative $\bpsi_{l+1}$ at level $l$}
      \ForAll {blocks $I^{(l)}_q$ at level $l$} \Comment{there are exactly $l$ blocks at level $l$}
      \State $i0 \gets \text{leftend} (I^{(l)}_q)$; \hspace*{0.5pc}
      $i1 \gets \text{rightend} (I^{(l)}_q)$  \Comment{$I^{(l)}_q = [i_0,i_1]$}
      \If{$(i_0 < i_1)$} \Comment{the block $I^{(l)}_q$ is not a leaf}
      \ForAll{$\ist \in\{i_0,\ldots,i_1\}$}
      \State $\bpsi_{l+1} \gets \text{buildvector} ([i_0,i_1],\ist)$ 
      \Comment{construct $\bpsi_{l+1}$ using \cref{from_l_to_plusone}}
      \State $\text{coeff}(\ist) \gets \lvert \ipr{\bpsi_{l+1}\bpsi_{l+1}^T}{\mA} \rvert^2$ 
      \State $\ist \gets \ist +1$ \Comment{update the index of the next
        tentative $\bpsi_{l+1}$}
      \EndFor \Comment{next index $\ist$ so that $[i_0,i_1] = [i_0,\ist] \cup
        [\ist+1,i_1]$}  
      \EndIf
      \EndFor \Comment{move to the next block at level $l$}
      \Statex \Comment{We have explored all the blocks at level $l$. We now find the block
        $[i_0,i_1]$ and the index $\ist$ of the split that result in the largest
        $\lvert \ipr{\bpsi_{l+1}\bpsi_{l+1}^T}{\mA}\rvert^2$. We save the corresponding  $\bpsi_{l+1}$
        in $\bPsi$}    
      \State $\big([i_0,i_1], \ist\big)\gets \argmax{i_0,i_1}\argmax{\ist \in \{i_0,\ldots,i_1\}}
      \big(\text{coeff})$
      \State $I_{q}^{(l+1)} \gets [i_0,\ist]$;  \hspace*{0.5pc}
      $I_{q+1}^{(l+1)} \gets [\ist+1,i_1]$
      \State $\bpsi_{l+1} \gets \text{buildvector} \big(I_{q}^{(l+1)}, I_{q+1}^{(l+1)}\big)$
      \Comment{use \cref{from_l_to_plusone} to construct $\bpsi_{l+1}$ }
      \State $\bPsi(:,l+1) \gets \bpsi_{l+1}$ \Comment{add $\bpsi_{l+1}$ to the Soules basis}
      \EndFor \Comment{go down to a finer level}
      \State \Return $\bPsi$ \Comment{return the Soules basis}
      \EndProcedure
    \end{algorithmic}
  \end{algorithm}
  \subsection{Spectral clustering of the nodes.}
  The greedy construction of the Soules basis necessitates that $\mA$ be ``well-aligned'' --
  in the sense that nodes are aggregated in clusters wherein the local node connectivity is
  approximately constant. This required step is equivalent to the estimation of a piecewise
  constant graphon for each observed graph in \cite{han22}. We use a spectral clustering
  method based on the eigenvectors of the normalized graph Laplacian
  \cite{damle19,meyer2014,Shen08} to organize nodes into clusters; the algorithm only
  requires the knowledge of the number $M$ of clusters. This prerequisite step only provides a grouping of the nodes based on
  the connectivity measurements (see \cref{clustering}). After this coarse clustering, the
  adjacency matrices $\bA^{(1)},\ldots,\bA^{(T)}$ are not aligned: one cannot match the
  nodes from one graph to another. In short, this step corresponds to the approximation of
  each $\bAt$ using a step graphon (e.g.,
  \cite{airoldi13,borgs20,ferguson23a,gao15,han22,lovasz12,olhede14,xu21}).
  As demonstrated in the experiments, the clustering of the nodes into communities is not
  always accurate. Fortunately, the choice of the Soules vectors in \cref{BSB} only relies
  on the coarse scale eigenvectors. These eigenvectors are determined by computing $\lvert
  \ipr{\bpsi_{l} \bpsi_{l}^T}{\mA}\rvert^2$. The support of $\bpsi_k \bpsi_k^T$ is large
  for $k=2, 3, \ldots$, and $\bpsi_k \bpsi_k^T$ is piecewise constant.  The random noise
  in $\sE{\pr}$ -- created by the incorrect assignment of nodes to the wrong community --
  is partly suppressed when computing $\big \lvert \ipr{\bpsi_k \bpsi_k^T}{\sE{\pr}}\big
  \rvert^2$. \cref{clustering} demonstrates the coarse clustering (see
  \cref{clustering}-right) applied on a randomly permuted adjacency matrix (see
  \cref{clustering}-center) of the $\sbm{\bp}{q}{n}$ model (see \cref{clustering}-left)
  
  The inherent uncertainty associated with the cluster labels is resolved by ranking the
  clusters according to their volume. This spectral clustering only requires that we compute
  the $M$ dominant eigenvectors of the sample symmetric normalized adjacency matrix,
  $\sE{\hA}$, and therefore does not increase significantly the computational load of the
  algorithm.


  
    \begin{figure}[H]
      \centerline{
        \includegraphics[width=0.3\textwidth]{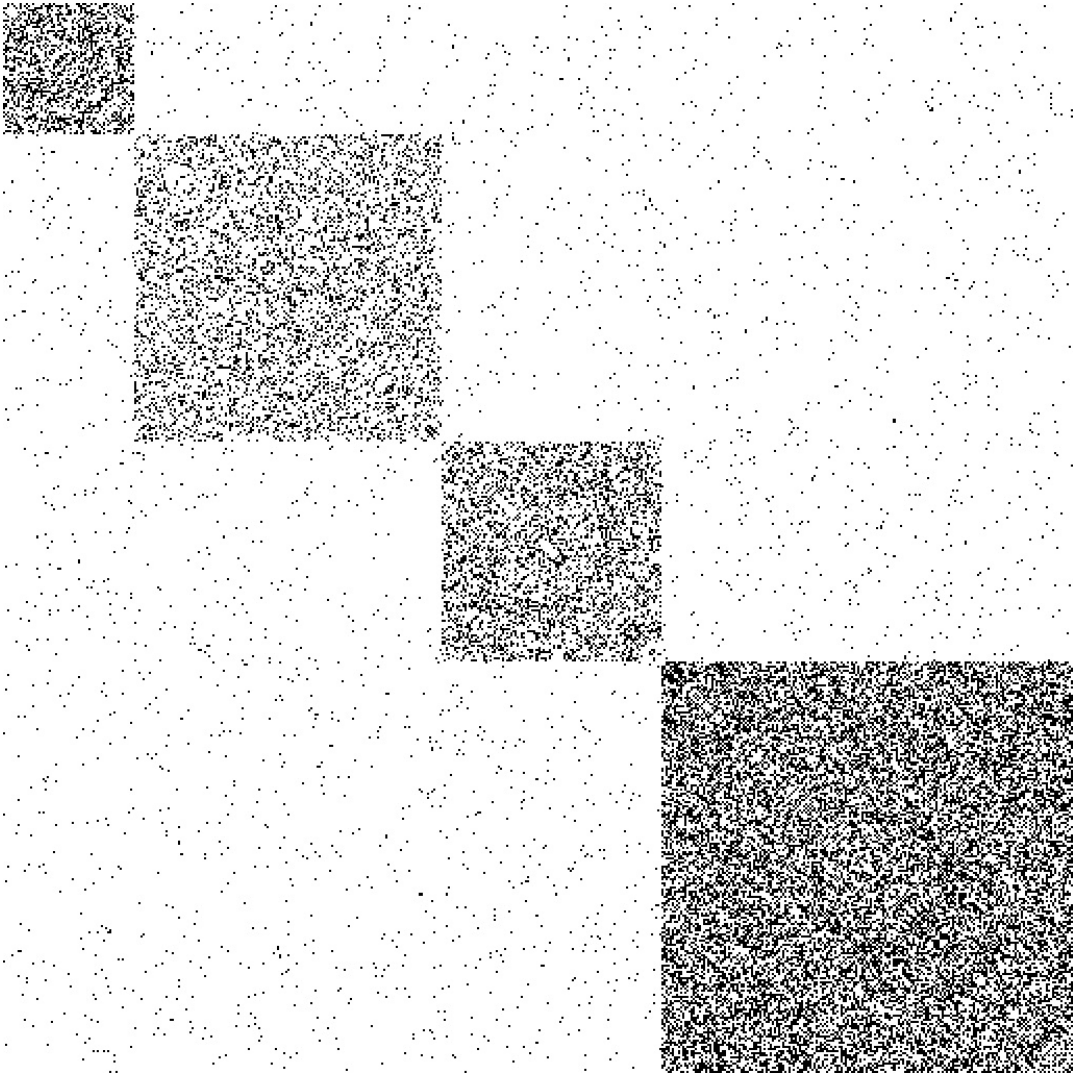} 
        \raisebox{5pc}{\includegraphics[width=1.5pc]{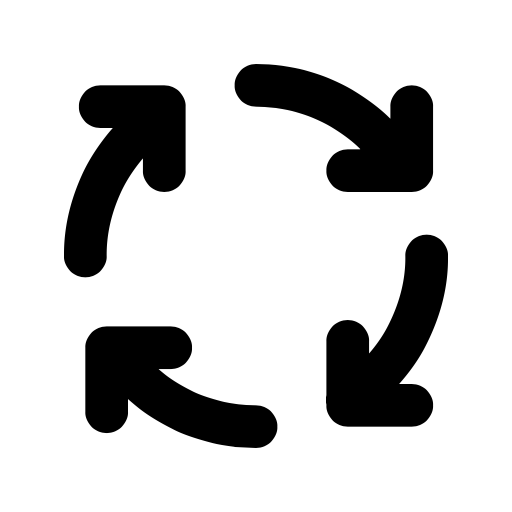}}
        \includegraphics[width=0.3\textwidth]{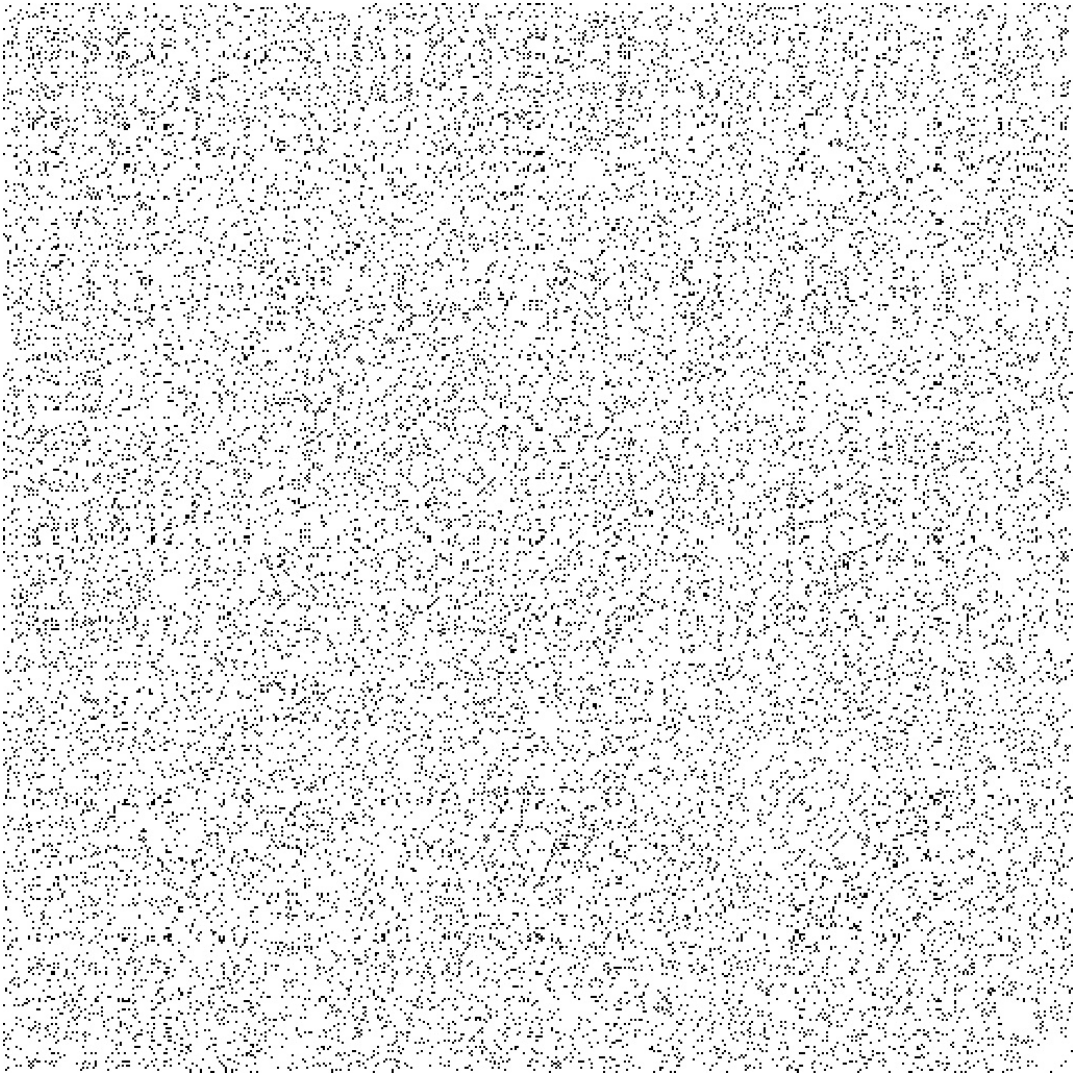}
        \raisebox{5pc}{\includegraphics[width=1.5pc]{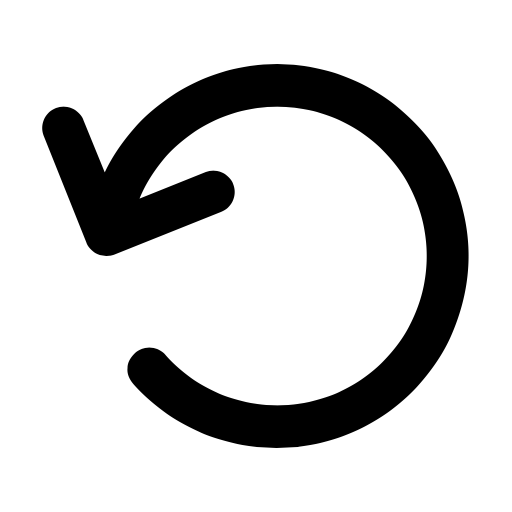}}
          \includegraphics[width=0.3\textwidth]{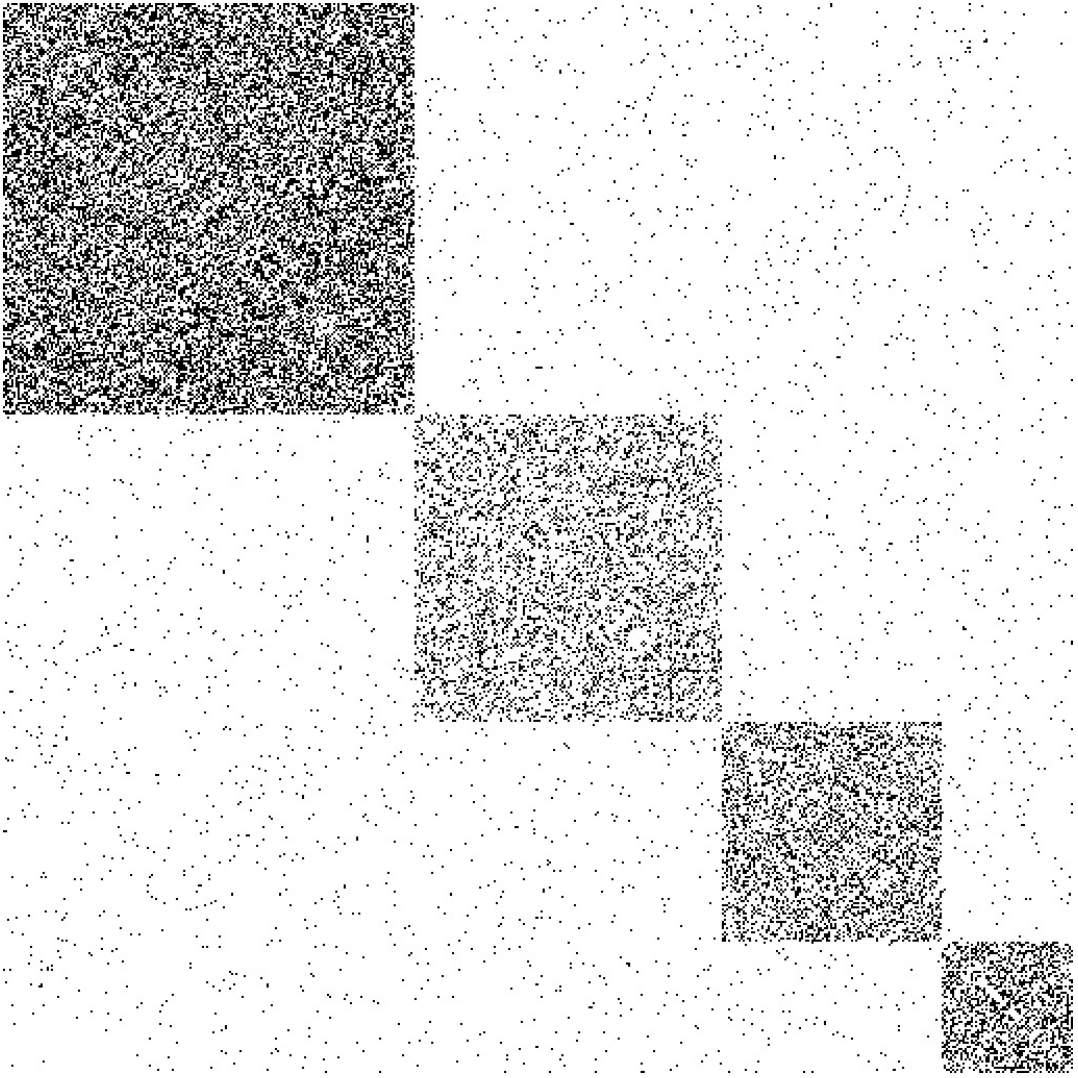} 
      }
      \caption{ Left: the adjacency matrix of a random realization of the
        $\sbm{\bp}{q}{n}$ model. We have $M=4$ communities of sizes $63, 147, 105,  197$;
        the graph size is $n=512$; the edge probability within community $i$ was $p_i \propto
        (\log{n})^2/n$, The edge probability across communities was $q=2 (\log{n})/n$. Center:
        the adjacency matrix is permuted with a random permutation.  Right: coarse clustering
        of the randomly permuted matrix. This last adjacency matrix is the input to
        \cref{BSB}; see \cref{les-soules}-left.
        \label{clustering}}
    \end{figure}
  
  \noindent
    
    \subsubsection{Algorithmic Complexity}
    The optimization problem given by \cref{le_gros_probleme} cannot be solved with a
    brute-force search, since we search $\bPsi$ over the orthogonal group
    $\OG(n)$. Polynomial-time numerical methods \cite{boumal14} have been devised to take
    advantage of the lower dimension of this Riemannian submanifold. In comparison, the
    computational complexity of \cref{BSB} is linear in $n$. We need to process $M$ levels;
    the discovery of $\bpsi_m$ at each level $m$ requires that we evaluate $\bigg \lvert
    \ipr{\bpsi_{m+1}\bpsi_{m+1}^T}{\mA} \bigg \rvert^2$ $n$ times, for an overall complexity
    of $M n$.
  
  \subsection{Theoretical guarantees for the algorithm.}
  Our analysis of \cref{BSB} is performed under the assumption that the input to the
  algorithm is not the sample mean adjacency matrix $\mA$ but its population equivalent $\bP
  = \E{\pr}$. Our experiments (see \cref{effet-des-parametres}-left) confirm the validity of
  this assumption. A finite sample analysis of the error bounds is left for future work.
  The following lemma proves that the first $M$ vectors of the Soules basis estimated by
  \cref{BSB} solve \cref{le_gros_probleme3}. Unlike \cref{lemmaLowe24}, \cref{lemma4} holds
  for all SBM, not just balanced SBM.
  \begin{lemma}
    \label{lemma4}
    Let $\bP$ be the population mean adjacency matrix of $\sbm{\bp}{q}{n}$ defined by
    \begin{equation}
      \bP = \sum_{m=1}^M (p_m - q)\one_{B_m}\one_{B_m}^T + q \bJ,
      \label{SBM}
    \end{equation}
    where the $M$ blocks $\{B_m\}$ form a partition of $[n]$. Let $J_l, 1 \le l \le M$ be
    the leaves in the binary Soules tree (these are the $M$ intervals that are no longer split,
    see \cref{the-soules-tree}-right) after $M$ steps of \cref{BSB}. Then, the $M$ blocks
    $\{B_m\}$ in \cref{SBM} coincide with $M$ the intervals $\{J_l\}$. The entries of the
    matrix $\sum_{k=1}^M \bpsi_k \bpsi_k^T$ satisfy
    \begin{equation}
      \sum_{k=1}^M \bpsi_k \bpsi_k^T(i,j) =
      \begin{cases}
        \displaystyle \frac{1}{\lon{J_m}} & \text{if} \mspace{8mu} (i,j) \in J_m \times J_m,\\
        0 & \text{otherwise},
      \end{cases}
    \end{equation}
    where $\lon{J_m}$ is the length of the interval $J_m$.  
  \end{lemma}
  The proof of \cref{lemma4} can be found in \cref{proof-lemma4}.  The proof relies on two
  different facts. The first shows that a top-down exploration of the Soules binary tree,
  when the first Soules vector is $\bpsi_1 = \scalebox{0.8}{${n}^{-1/2}$} \one$, always
  creates a matrix $\bE_M = \sum_{k=1}^M \bpsi_k \bpsi_k^T$ that is piecewise constant on
  square blocks aligned along the diagonal, and zero outside of the blocks (see
  \cref{corollaryp37}). This property only relies on the fact that the sequence of
  $\{\bpsi_k\}_{1 \le k \le M}$ have nested supports.

  The second result specifically addresses the construction of each $\bpsi_k$ in
  \cref{BSB}. We prove in \cref{corollary3p51} that at each level $k$, the 
  Soules vector $\bpsi_k$ returned by \cref{BSB} is aligned with the boundary of a
  block $B_m$ of the edge probability matrix $\bP$. At level~$M$, \cref{BSB} has
  discovered all the $M$ blocks.
  \begin{corollary}
    Let $\bPsi \eqdef \begin{bmatrix}\bpsi_1  \cdots  \bpsi_n \end{bmatrix}$ be the Soules
    basis returned by \cref{BSB}. Then $\bPsi$ solves \cref{le_gros_probleme3}.
    
  \end{corollary}
  \noindent After \cref{BSB} returns $\begin{bmatrix}\bpsi_1 \cdots \bpsi_n \end{bmatrix}$ of
  $\Lmu$, we reconstruct $\fmpr$ using \cref{from_Laplacian_to_barycentre}.
  \section{The reconstruction of \texorpdfstring{$\fmpr$}{TEXT} \label{la_section_des_details}}
  
     By virtue of being normalized, $\cL$  cannot reveal the corresponding
    adjacency matrix $\bA$. Indeed, from \cref{Ahat}, we have 
    \begin{equation}
      \bA = \bD^{1/2} \big(\id - \cL \big)\bD^{1/2}, \label{fromLtoA}
    \end{equation}
  
  and therefore we need the knowledge of the degree matrix $\bD$ to
  recover $\bA$. We describe in the next paragraph an estimator $\hD$ of the degree
  matrix of $\fmpr$.
  \subsection{Estimation of the degree matrix of \texorpdfstring{$\fmpr$}{TEXT}}
  We observe that \cref{lemma4} yields an estimate of the location of the blocks $B_m \times
  B_m$, $1 \le m \le M$ associated with $\sbm{\bp}{q}{n}$. An estimate of the degree matrix,
  $\hD$, is obtained by averaging the degrees of all the nodes in each block $B_m$ of $\mA$,
  \begin{equation}
    \widehat{d}_m \eqdef \sum_{(i,j)\in B_m \times B_m} \big[\mA\big]_{ij}, \mspace{16mu}
    1 \le m\le M. 
    \label{les-degres}
  \end{equation}
  \noindent A quick calculation confirms that the estimator $\widehat{d}_m$ concentrates
  around the ``within block'' expected degree $\dbar_m$ (see
  \cref{expected_normalized_Laplacian}). We have
  \begin{equation}
    \widehat{d}_m = T^{-1} \sum_{t=1}^T \sum_{(i,j)\in B_m \times B_m} a^{(t)}_{ij}.
  \end{equation}
  Therefore $\widehat{d}_m$ is the sum of $T$ independent binomial random variables, and it
  concentrates around its mean $\dbar_m$. The variation of $\widehat{d}_m$ around $\dbar_m$
  is bounded by Hoeffding inequality,
  \begin{equation}
    \forall m \in [M],  \forall \; T \ge 1,\forall \; \delta >0,
    \prob{\Big \lvert  \widehat{d}_m - \dbar_m \Big \rvert  \ge \delta
    }
    \le
    \exp{\displaystyle \big(-4T\delta^2/(\lon{B_m}(\lon{B_M} -1))\big)}.
  \end{equation}
  
     We recall that $\dbar_m \eqdef \E{d_{ii}} = \lon{B_m}\big\{p_m +
    (n/\lon{B_m} -1)q\big\}$. Since we always have $p_i \gg q$ (when $p_i \approx q$,
    communities can no longer be detected \cite{abbe18,mossel14}), we can neglect $q$ in the
    estimation of $\dbar_m$. Wherefore we conclude that $\widehat{d}_m$ is asymptotically
    unbiased when the sample size $T\rightarrow +\infty$.
  
  \subsection{An estimator of the adjacency matrix of the barycentre graph}
  We conclude from \cref{from_Laplacian_to_barycentre} that we have
  \begin{equation}
    \fmprM \eqdef \hD^{1/2} \big(\id - \hLM \big)\hD^{1/2},
    \label{le_graph_barycentre}
  \end{equation}
  where $\hLM$ is given by \cref{truncated_Laplacian}, and $\hD$ is given by
  \cref{les-degres}.

  \begin{figure}[H]
    \centerline{
      \includegraphics[width=0.275\textwidth]{A}
        \raisebox{5pc}{\includegraphics[width=1.5pc]{rotate-left}}
      \includegraphics[width=0.275\textwidth]{EA-aligned} 
        \raisebox{5pc}{\includegraphics[width=1.5pc]{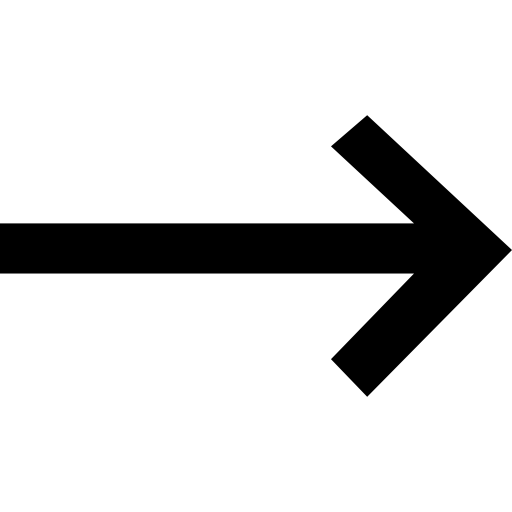}}
      \raisebox{-0.25pc}{\includegraphics[width=0.3125\textwidth]{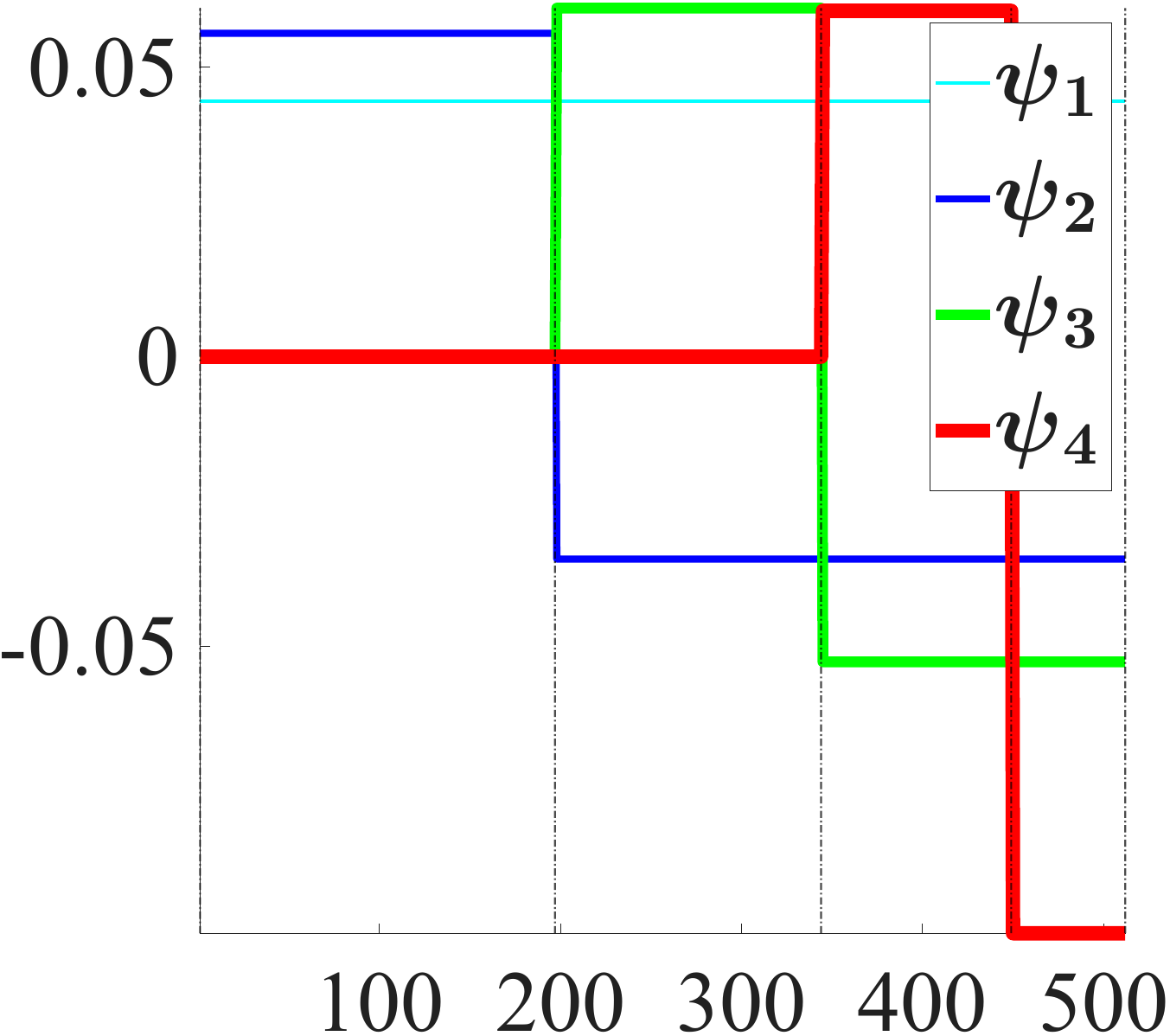}}
    }
    \caption{Left: the realization of $\sbm{\bp}{q}{n}$ shown in \cref{clustering} after a
      random permutation. Center: coarse clustering of the randomly permuted matrix.  After
      grouping nodes into homogeneous clusters, we have $M=4$ communities of sizes $197,
      105, 147, 63$; the graph size is $n=512$. The first four trivial Soules vectors
      accurately detected the boundaries between the blocks (see bars indicating the edge
      boundaries), in spite of the very low contrast between the communities (see
      \cref{clustering}-left).
      \label{les-soules}}
  \end{figure}
  \section{Experiments
    \label{experiments}}
  \subsection{Random graph models}
  We compare our theoretical analysis to finite sample estimates, which were computed using
  numerical simulations. The software used to conduct the experiments is publicly available
  \cite{meyer25}. All graphs were generated using the $\sbm{\bp}{q}{n}$ model. The nodes of
  the random realizations of the adjacency matrices are permuted with a different random
  permutation for each realization (e.g., see \cref{clustering}-centre). 
  \subsubsection{Experimental validation of \cref{lemma4}}
  
     The first experiment provides a validation of \cref{lemma4}. For this
    experiment, we use a regular $\sbm{\bp}{q}{n}$ that is not balanced. The number of
    communities is similar to the numbers that are used in the literature (e.g.,
    \cite{airoldi13,dolevzal21,gao15,gerlach18,olhede14}). Because the $\sbm{\bp}{q}{n}$ is
    not balanced, we are effectively testing \cref{BSB} outside of the theoretical
    guarantees provided by \cref{lemma4}.
  
  We use $M=4$ communities of sizes $63, 147, 105, 197$ (see \cref{clustering}-left). The edge
  probabilities were given by $p_i = c_i {\log{n}}^2/n$, where the scaling factor $c_i$ was
  chosen randomly in $[1,4]$, and $q=2 \log{n}/n$. These edge probabilities yield sparse
  graphs that are connected almost surely. We generate a randomly permuted single realisation
  of the SBM ($T=1$) (see \cref{clustering}-center, and \cref{les-soules}-left).
  
     We first illustrate the selection of the Soules vectors (guaranteed by
    \cref{lemma4}). The case $T=1$ is the least favorable scenario, since we do not expect
    that the sample mean adjacency matrix $\mA$ be close to the population mean adjacency
    matrix $\E{\pr}$. In addition, all the blocks have different sizes (the SBM is not
    balanced), and the hypotheses of \cref{lemma4} do not hold.  Whence we expect that the
    estimation of the Soules basis to be most challenging.
  As shown in \cref{les-soules}-right, the first three non trivial Soules vectors accurately
  detected the boundaries between the blocks (vertical bars located at $i = 63, 210, 315$
  mark the block boundaries), in spite of the very low contrast between the communities (see
  \cref{clustering}-left). This numerical evidence supports the theoretical analysis of
  \cref{lemma4}.
  We then evaluated the accuracy of
  \cref{le_graph_barycentre}. \cref{les-reconstructions}-left displays the original edge
  probability matrix $\bP=\E{\pr}$; \cref{les-reconstructions}-center displays the adjacency
  matrix of the barycentre graph $\fmprM$ using the top $M=4$ Soules vectors, and
  \cref{les-reconstructions}-right displays the residual error $\E{\pr} - \fmprM$. The mean
  squared error, defined by
  \begin{equation}
    \label{MSE}
    \mse{\E{\pr} - \fmprM} \eqdef \frac{1}{n^2}
    \sum_{i=1}^n \sum_{j=1}^n \Big \lvert p_{ij} - \widehat{p}_{ij} \Big \rvert^2,
  \end{equation}
  was $3.0484e-05$.

  \begin{figure}[H]
    \centerline{
      \includegraphics[width=0.3\textwidth]{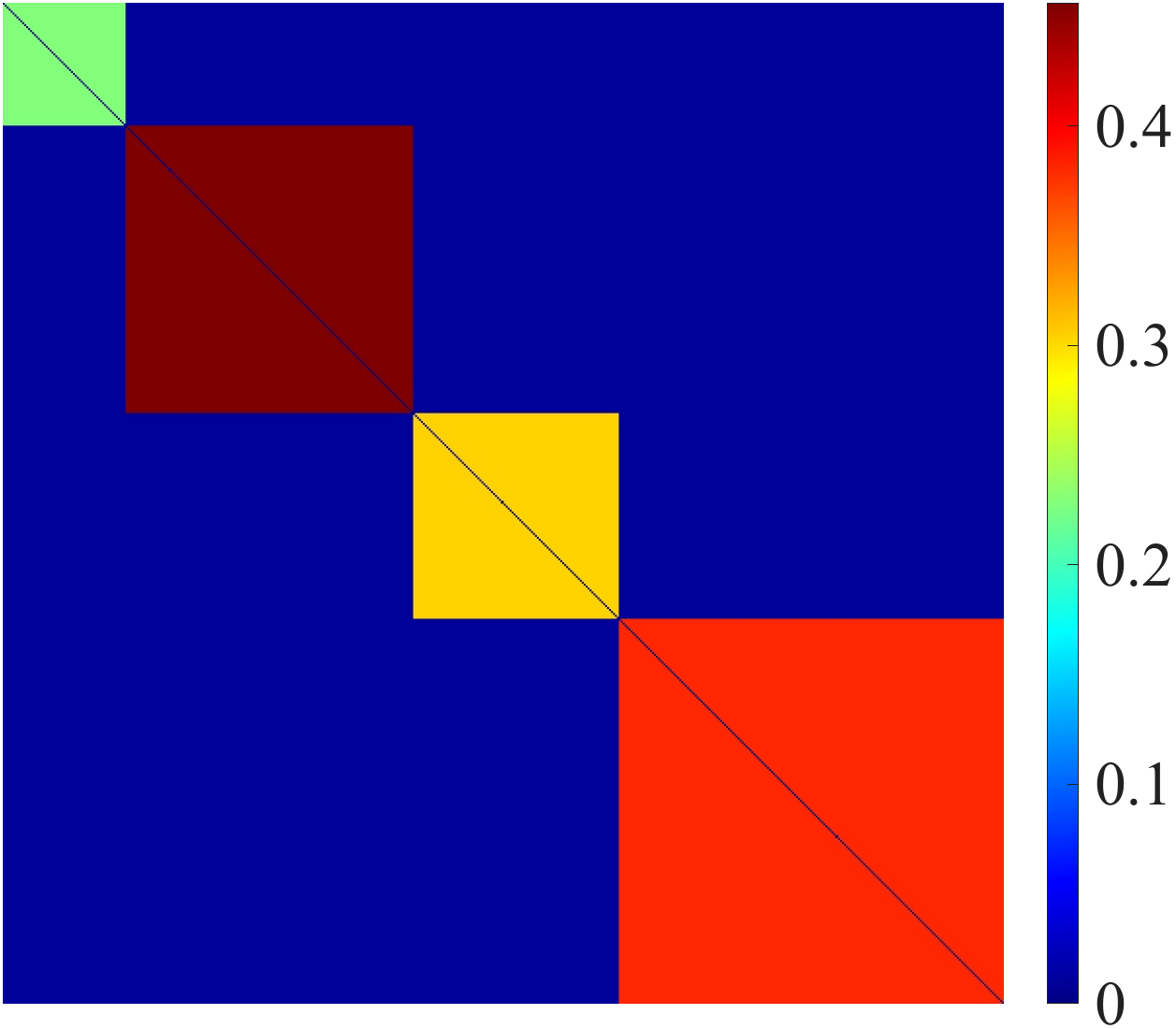}
      \includegraphics[width=0.3\textwidth]{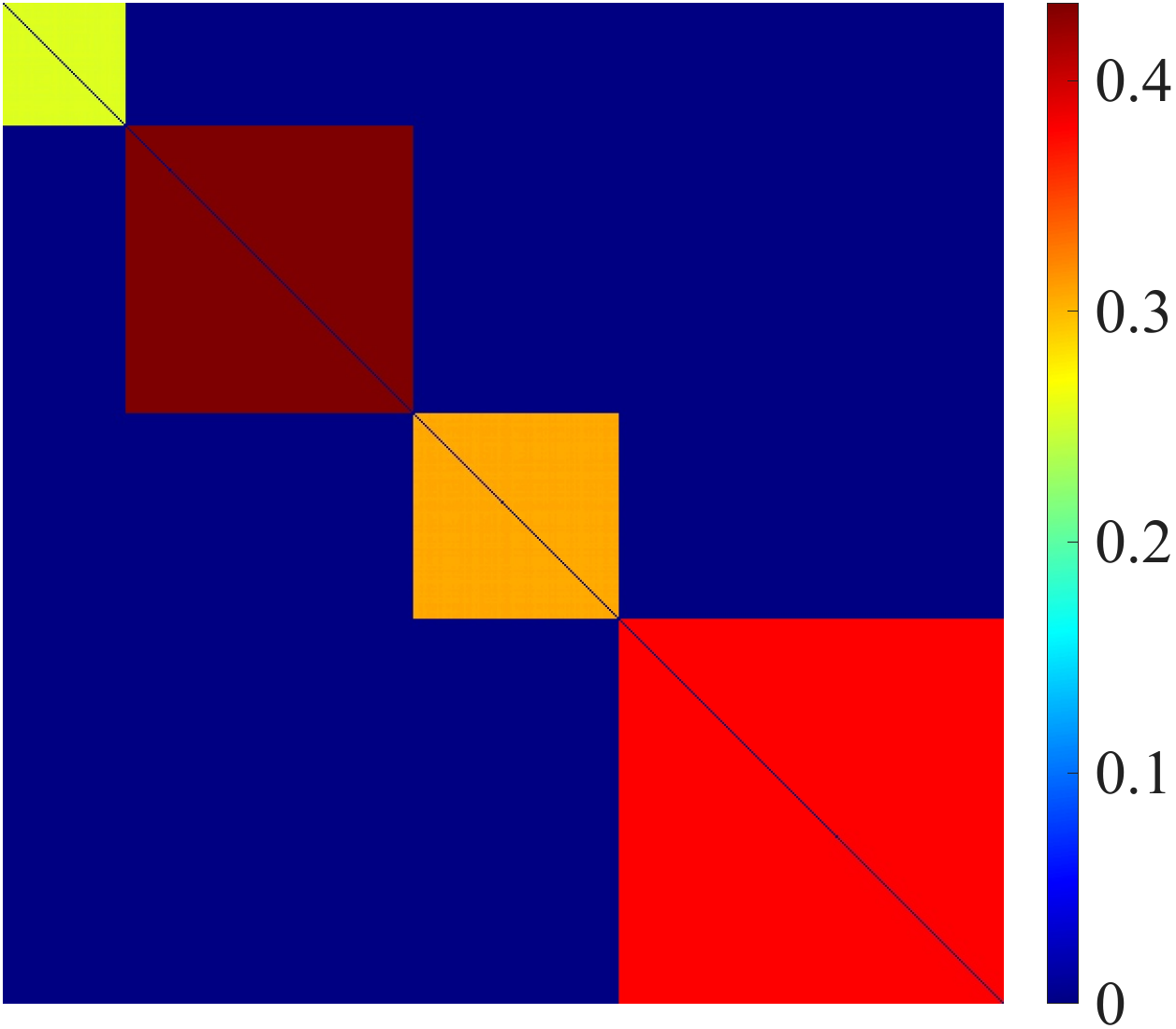}
      \includegraphics[width=0.315\textwidth]{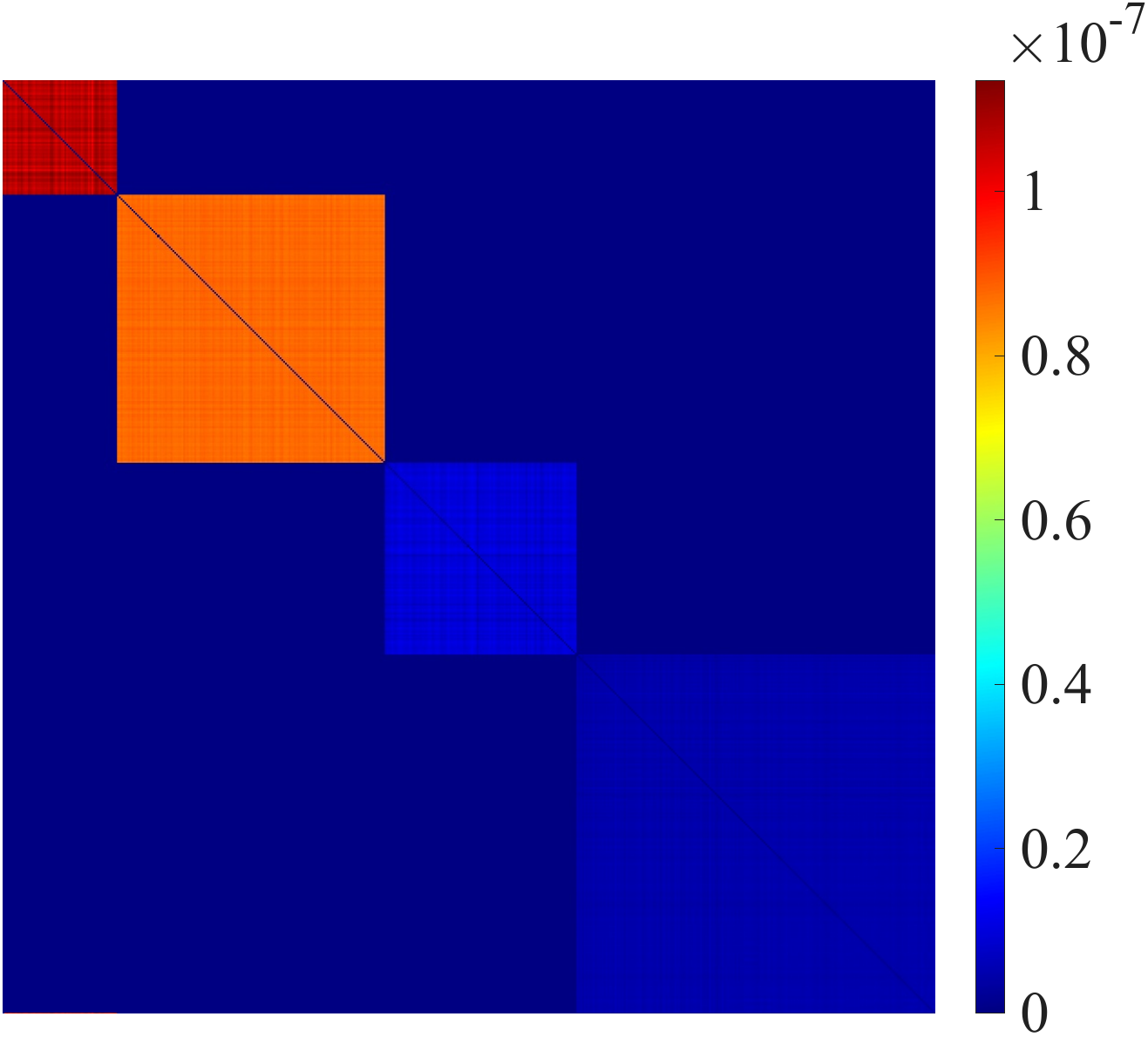}
    }
    \caption{Left: the edge probability matrix $\E{\pr} = \bP$ associated with
      $\sbm{\bp}{q}{n}$ that was used to generate the adjacency matrix in
      \cref{clustering}-left; center: the barycentre graph $\fmpr$, given by
      \cref{le_graph_barycentre}; right: the residual error between $\bP$ and $\fmpr$.
      \label{les-reconstructions}}
  \end{figure}
    \subsubsection{Rate of convergence of \texorpdfstring{$\fmprM$}{TEXT} as a function of the graph
    size} \cref{leprobleme-whp} promises that $\Lmu$ converges to $\cLbar$ with probability
  converging to $1$ as the graph size $n\rightarrow +\infty$. Therefore, we studied the
  effect of the graph size, $n$, on the mean squared error (see
  \cref{effet-des-parametres}-left).  We rescaled the $M=4$ $\sbm{\bp}{q}{n}$ model
  described in the previous paragraph, keeping the relative sizes of the communities the
  same, and increased the graph size from $n=100$ to $n=1,075$. For each value of $n$, we
  computed the mean squared error (\cref{MSE}). As expected, the error decreases as a
  function of $n$.

  We found $\mse{\E{\pr} - \fmprM} \propto n^{-1.84}$. This rate of convergence is of the
  same order as the optimal (minimax) rate for the estimation of graphons under the mean
  squared error \cite{gao15,olhede14,xu18}. Since $\fmprM$ is a stochastic block model we do
  not expect the underlying graphon to be smooth, and therefore the optimal rate is the
  bound $n^{-1}\log(M) + n^{-2} M^2$ \cite{gao15,olhede14,xu18}.
  
  We note that the alignment performed by the spectral clustering is not always accurate
  (as reflected by the presence of outliers in the plot of the mean squared error; e.g.,
  $n = 374$ in \cref{effet-des-parametres}-left). This is due to the fact that we use a
  simple
  \begin{figure}[H]
    \centerline{
      \includegraphics[width=0.375\textwidth]{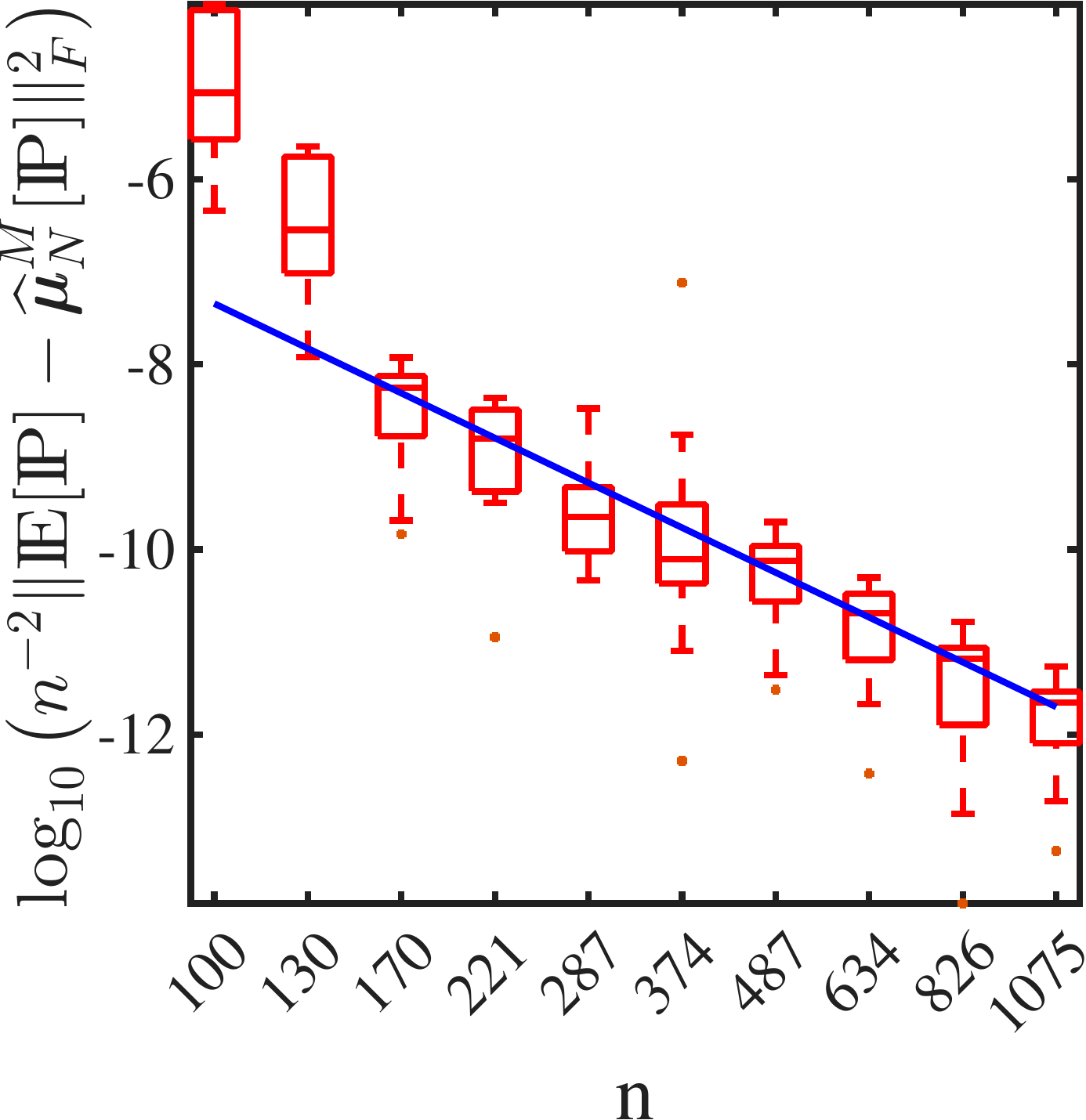}
      \hspace*{2pc}
 \raisebox{4pt}{\includegraphics[width=0.375\textwidth]{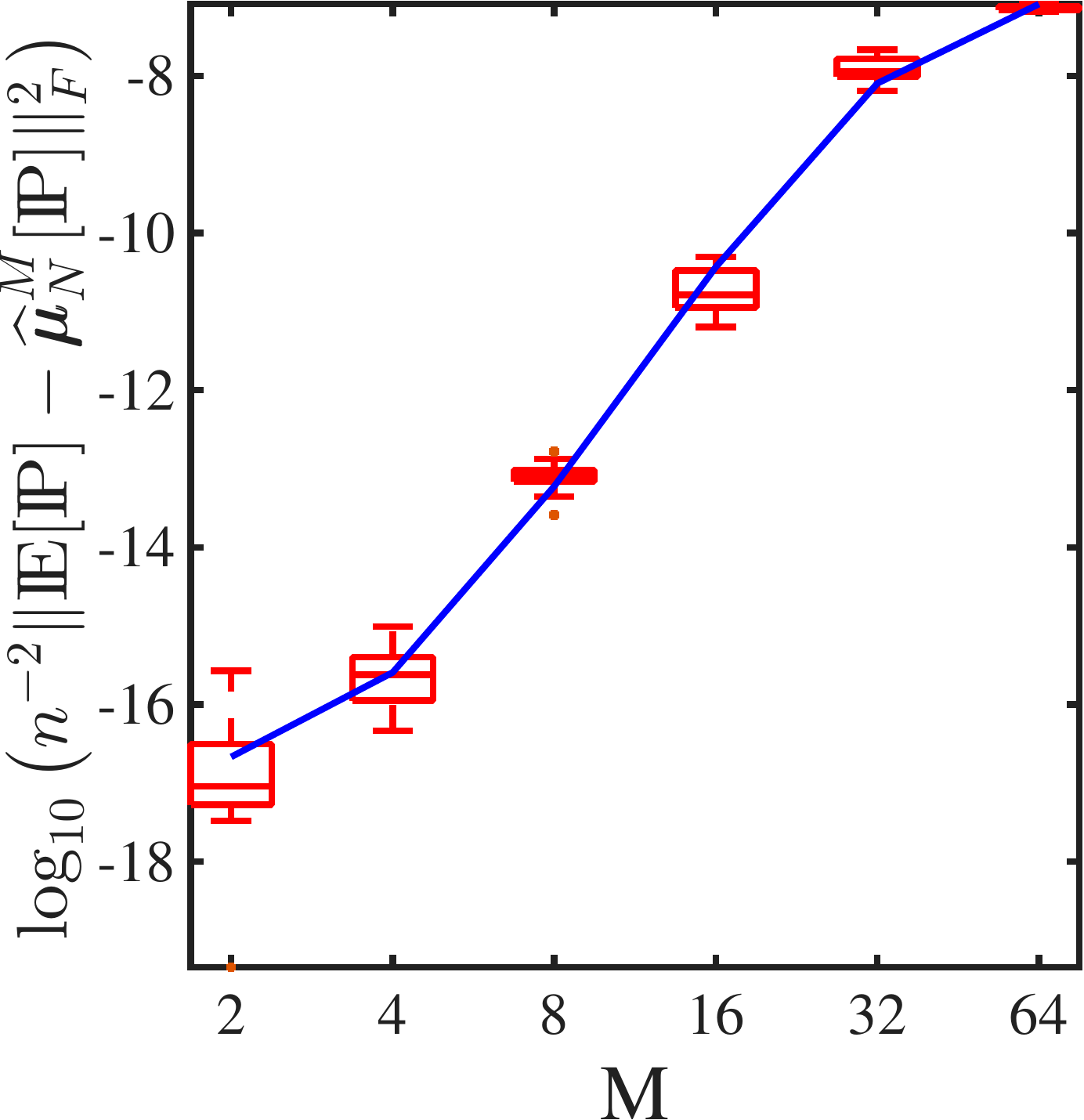}}
    }
    \caption{Left: mean squared error $\mse{\E{\pr} - \fmprM}$ as a function of the
      graph size, $n$. The graph is composed of $M=4$ communities, and is a scaled
      version of the graph shown in \cref{clustering}-left. Right: mean squared error
      $\mse{\E{\pr} - \fmprM}$ as a function of the number of blocks, $M$. Each graph is
      sampled from a balanced $\sbm{p}{q}{n}$ with $M$ blocks of size $n/M$; $p_i =
      3(\log{n})^2/n, q =2\log{n}/n$, and $n=1,024$.
      \label{effet-des-parametres}}
  \end{figure}
  \noindent spectral clustering algorithm. The present work focuses on the construction of
  the barycentre graph; the study of the combined performance of the pre-processing step
  with the computation of the barycentre is left for future work.
  \subsubsection{Effect of the number of blocks $M$}
  The next experiment illustrates the effect of the number of blocks $M$ in a balanced
  $\sbm{\bp}{q}{n}$ when the edge probabilities are equal, $p_1 = \cdots = p_M$. When $M$
  becomes large, then the first $M-1$ non trivial eigenvalues $\lambda_m$ of $\cL$, converge
  to $1$. Because these eigenvalues are no longer separated from the bulk, the regularized
  reconstruction \cref{truncated_Laplacian} becomes numerically unstable, and the
  reconstruction error increases (see \cref{effet-des-parametres}-right).
  \subsection{Real world graphs \label{lespetitsfrancais}}
  We evaluated the performance of our algorithm on a time-sequence of social-contact graphs,
  collected via RFID tags in an French primary school \cite{Stehle2011}. The dataset was
  described earlier in \cref{SBM-section}. In this dataset, the time-varying graphs undergo
  significant structural changes as they evolve in time (e.g., merging of two classes, or
  the emergence of a single subgraph as a connective hub between disparate regions of the
  graph (see \cref{lespetits}). 

  This dataset has been used as a standard benchmark to quantify the performance of
  sophisticated algorithms to analyse dynamic networks, which exhibit intricate dynamics
  and complex changes in connectivity and structural properties
  \cite{cencetti25,djurdjevac25,failla24,ferraz21,gauvin14,sattar23}.  For the purpose of
  this experiment, we think of each class as a community of connected students; despite
  the fact that classes are weakly connected (e.g., see \cref{lespetits} at times 9:00
  a.m., and 2:03 p.m.), the goal of the experiment is to recover the communities
  determined by the classes using the subset of graphs associated with the morning and
  afternoon periods separately.
  \noindent   We divide the school day into morning (8:30 AM--12:00 PM) and afternoon (2:00 PM --4:30
  PM). We exclude the lunch period because many students leave the school to take their lunch
  at home. The morning period is divided into $T=35$ time intervals of approximately 6
  minutes; the afternoon is divided into $T=26$ time intervals of approximately 6 minutes. For
  each time interval we construct an undirected unweighted graph $\Gt$, where the $n=232$
  nodes correspond to the students in the 10 classes.

  The morning barycentre graph is computed using the $T=35$ graphs associated with the
  morning period. The afternoon barycentre graph is determined using the $T=26$ afternoon
  graphs. We display in \cref{le_graph_du_matin}-left the graph
  \noindent   \begin{figure}[H]
    \centerline{
      \includegraphics[width=0.8\textwidth]{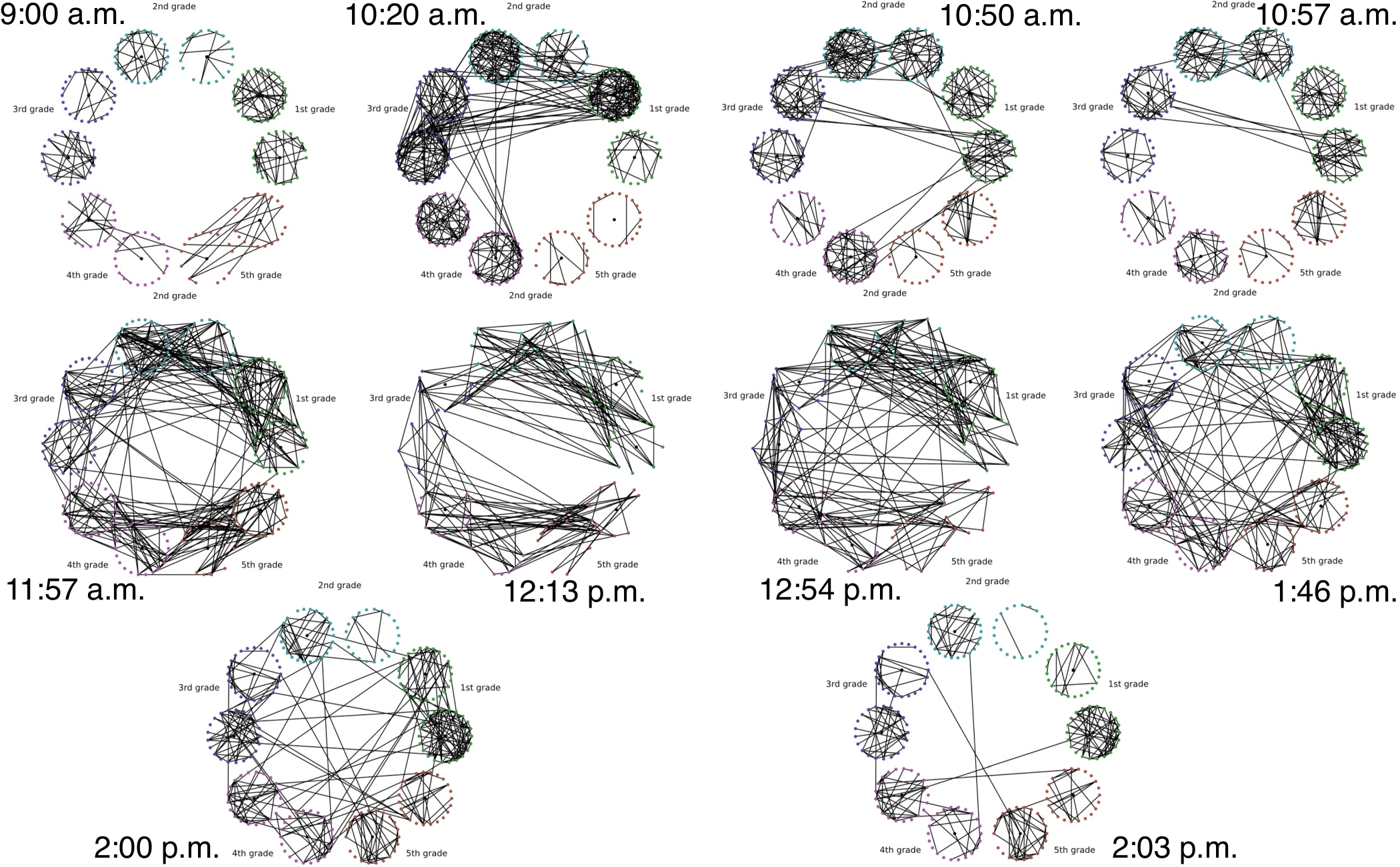}}
    \caption{Top to bottom, left to right: snapshots of the face-to-face contact graph at times
      (shown next to each graph) surrounding significant topological changes.
      \label{lespetits}}
  \end{figure}
  \begin{figure}[H]
    \centerline{
      \includegraphics[width=0.5\textwidth]{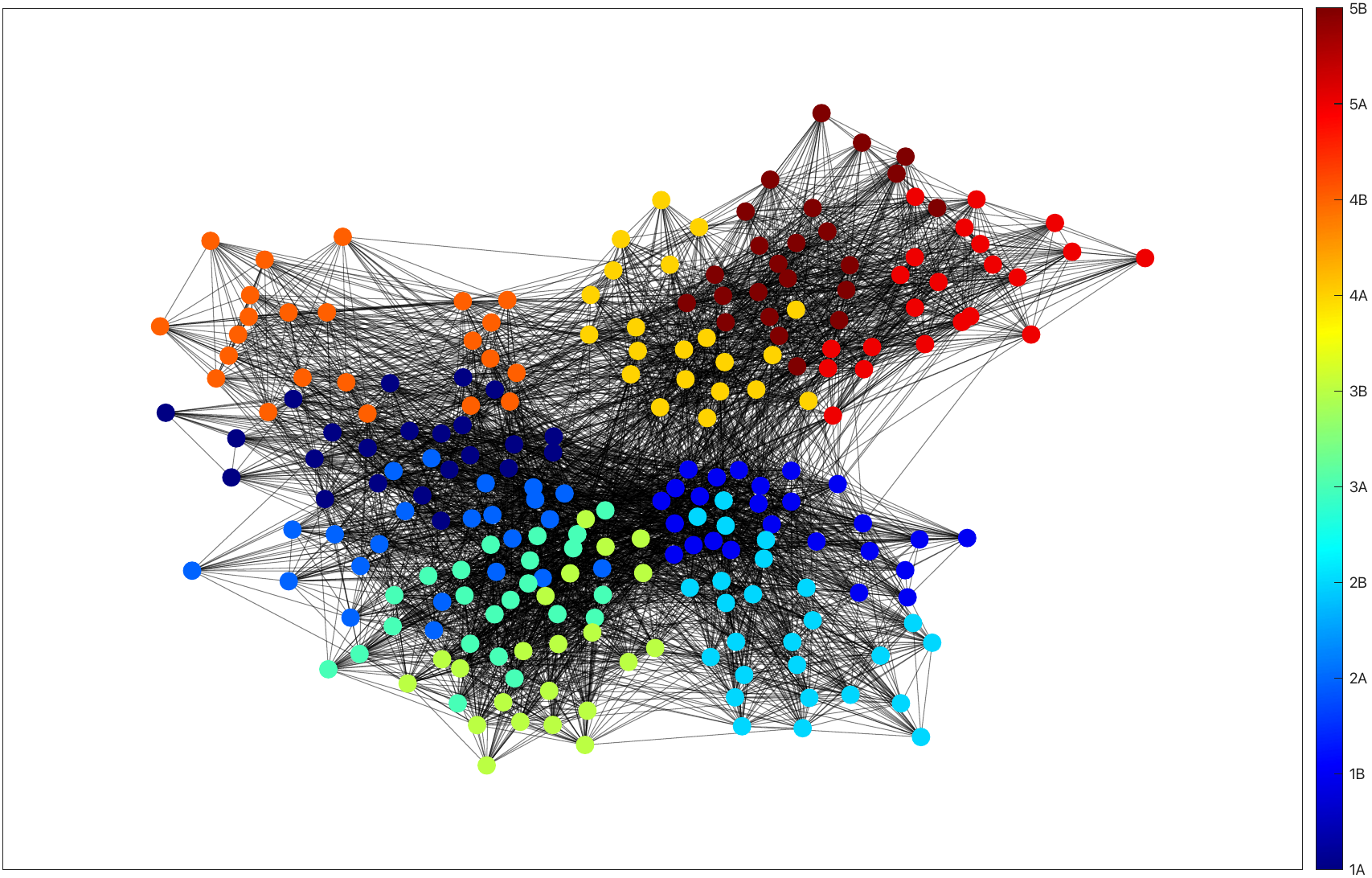}
      \includegraphics[width=0.5\textwidth]{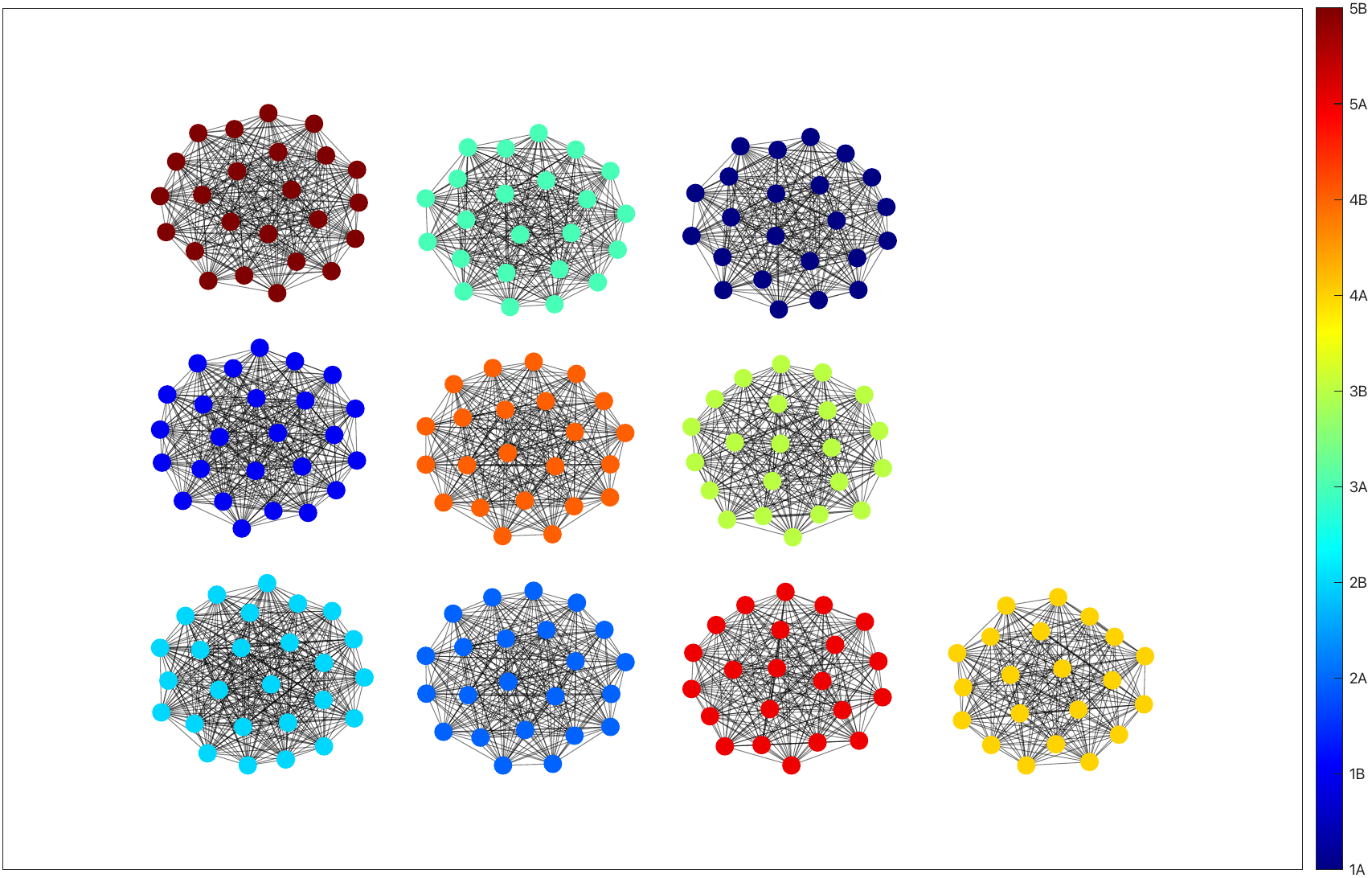}
    }
    \caption{Morning period. Left: graph of the sample mean adjacency matrix $\mA$ computed
      over the morning period; right: barycentre graph $\fmprM$ computed
      over the morning period.\label{le_graph_du_matin}}
  \end{figure}
  \begin{figure}[H]
    \centerline{
      \includegraphics[width=0.5\textwidth]{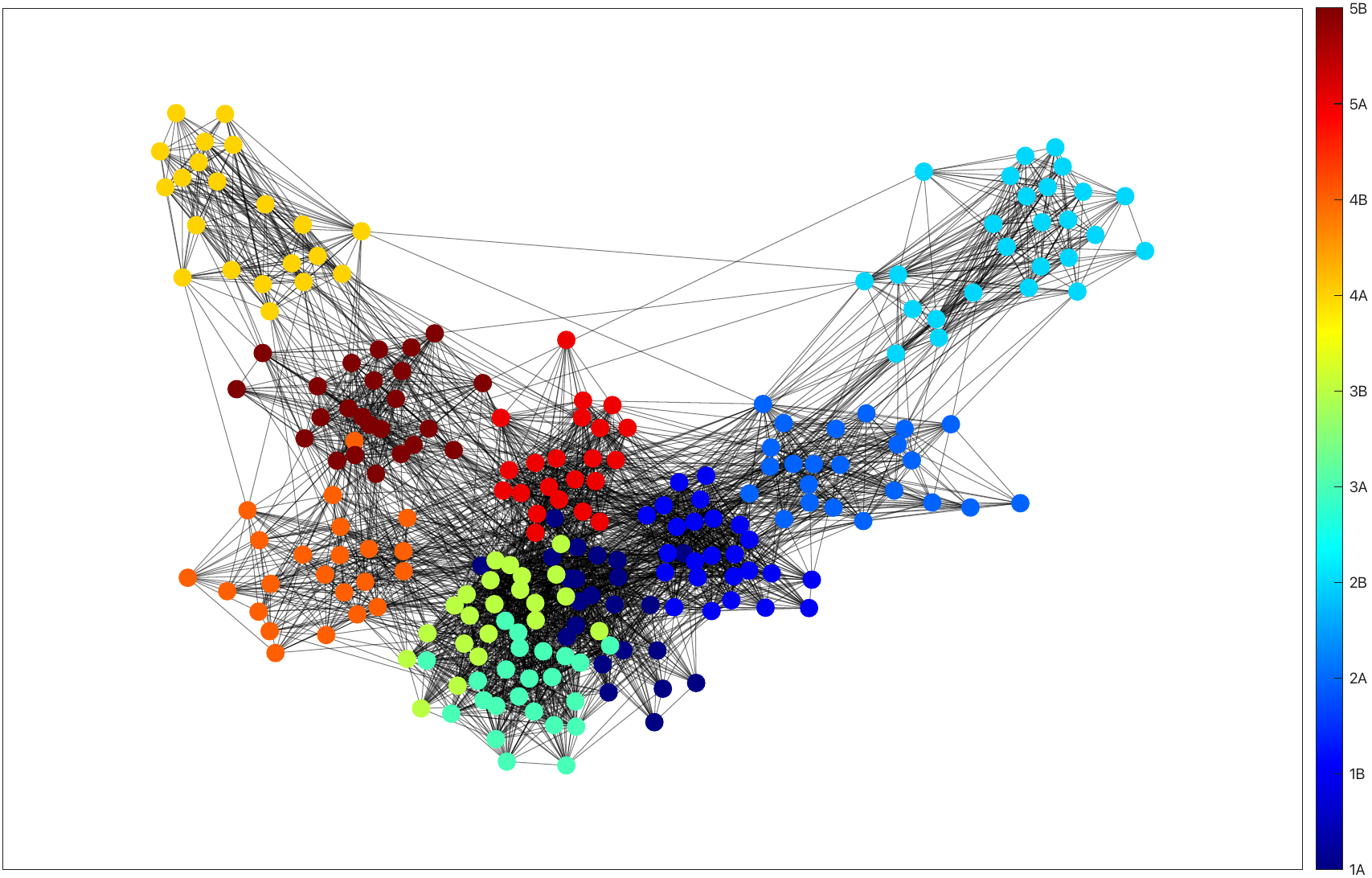}
      \includegraphics[width=0.5\textwidth]{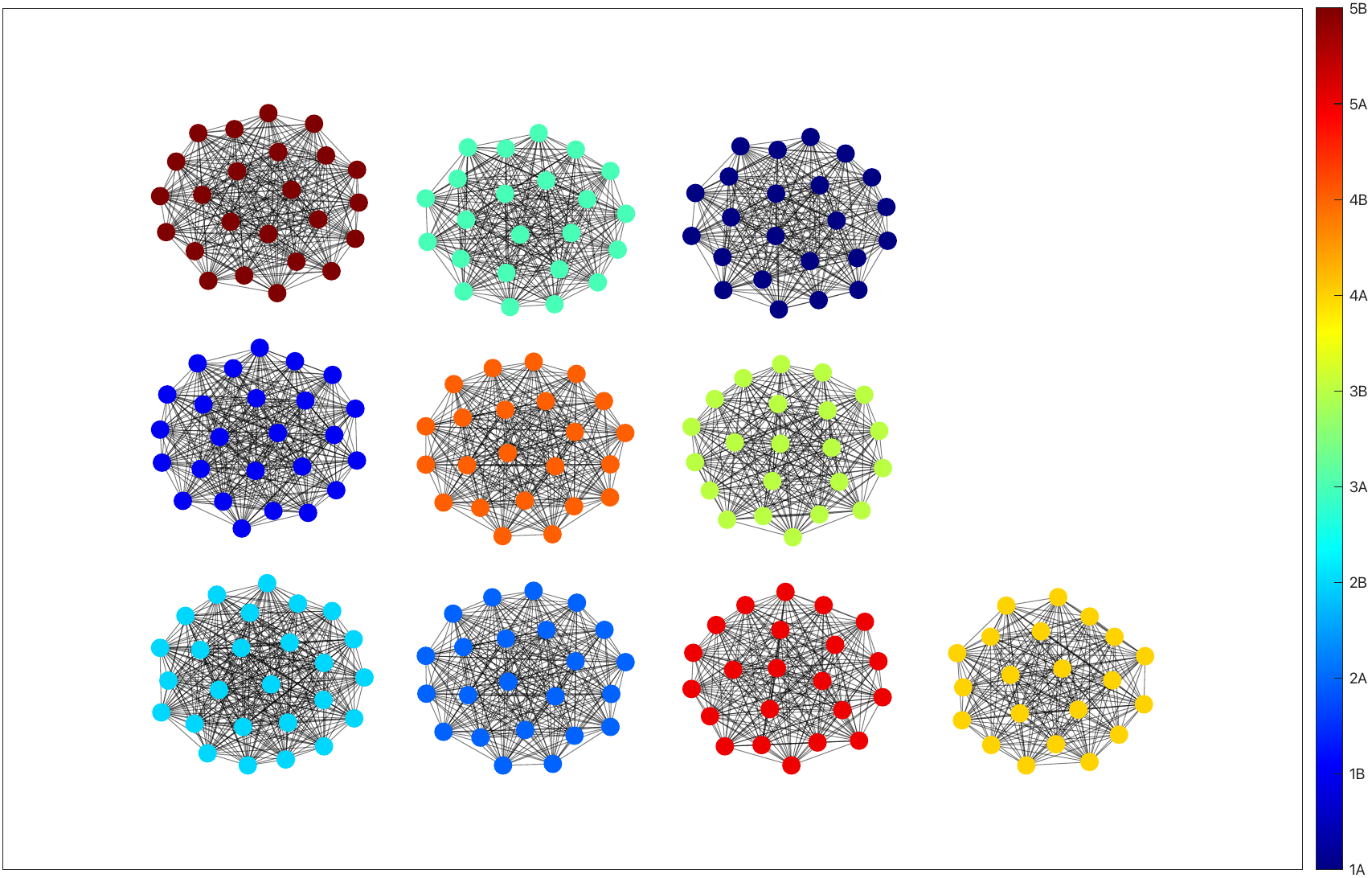}
    }
    \caption{Afternoon period. Left: graph of the sample mean adjacency matrix $\mA$
      computed over the afternoon period; right: barycentre graph $\fmprM$ computed over the
      afternoon period. \label{le_graph_de_lapresmidi}}
  \end{figure}
\noindent associated with the sample mean adjacency matrix $\mA$. We observe that the
events such as lunchtime and recess trigger significant increases in the number of links
between the communities, and disrupt the community structure (see \cref{lespetits}).
  \noindent As a result the community structure associated with the individual classes collapses in
  the graph constructed from $\mA$, both for the morning and afternoon periods (see
  \cref{le_graph_du_matin}-left, and \cref{le_graph_de_lapresmidi}-left). In comparison, the
  barycentre graph, which does not rely on the presence or absence of edges, is able to
  recover the individual classes (see \cref{le_graph_du_matin}-right, and
  \cref{le_graph_de_lapresmidi}-right). \Cref{les_valeurs_propres}, which was displayed in
  \cref{SBM-section}, displays the histogram of all the eigenvalues of $\cL$ integrated over
  the morning (left) and afternoon (right) periods. The ten lowest eigenvalues are separated
  from the bulk in the morning and afternoon distributions (see \cref{les_valeurs_propres}), which guarantees that the graph associated with
  these distributions will have a community structure. The separation is more noticeable
  during the morning since students spend more time in their classroom during this period.
  \section{Discussion
    \label{la-discussion}}
   
    In this work, we address the following question: when one computes a barycentre graph
    using a spectral pseudo-distance, one needs a basis of eigenvectors to reconstruct the
    barycentre graph: the distance only yields the eigenvalues. This paper contains several
    original contributions: (1) we formally define a criterion for estimating the
    eigenvectors, (2) we devise an algorithm to estimate the eigenvectors using the library of
    Soules vectors, (3) we provide a theoretical justification for the algorithm (in the
    context of the analysis of balanced Stochastic Block Models), and (4) we conduct
    experiments that suggest that the algorithm performs as expected beyond the ensemble of
    balanced stochastic block models. As many papers in the research area of data
    representation and analysis, this manuscript combines algorithms with theory, where the
    theory is sometimes developed in a more restricted framework.

    This work opens the door to further studies. Our experiments suggest that the method works
    beyond the restricted scope of balanced SBM, and one should be able to prove that
    \cref{BSB} solves \cref{le_gros_probleme} when graphs are sampled from
    $\sbm{\bp}{q}{n}$. A second extension concerns the application of \cref{BSB} to graphs
    that have different sizes. Because we work with the normalized graph Laplacian, one can
    always compare interpolate a discrete spectrum on the interval $[0,2]$, making the
    comparison of two graphs of different size possible. Finally, the analysis of the
    performance of the algorithm with graphs sampled from the preferential attachment models
    \cite{Yule1925,Barabasi1999} would have a significant practical impact. While less is
    known about the eigenvalues of the normalized Laplacian \cite{Flaxman05,Farkas2001}, the
    topology of the preferential attachment model (with the high-degree vertices) is
    frequently found in real world networks, and therefore this model is worthy of attention.
  
  \section*{Acknowledgments}
  The author is grateful to the anonymous reviewers for their insightful comments and
  suggestions that greatly improved the content and presentation of this manuscript.
  \section{Additional proofs
    \label{les-preuves}}

  \subsection{Proof of \cref{lemma4}
    \label{proof-lemma4}}
  We derive the proof of \cref{lemma4}. In the process, we prove several technical lemmata. 
  \subsubsection{The tensor product \texorpdfstring{$\bpsi_l\bpsi_l^T$}{TEXT}}
  \begin{lemma}
    \label{computation-psik}
    We choose $\bpsi_1 \eqdef \scalebox{0.8}{$\small n^{-1/2}$} \one$, and denote by
    $\bpsi_l$ the Soules vector returned by \cref{BSB} at level $l$. Let
    $\supp{\bpsi_l} = [i_0,i_1]$, and let $\ist$ be the location of the split in $[i_0,i_1]$
    such that $\rstr{\bpsi_l}{[i0,\ist]} > 0$ and $\rstr{\bpsi_l}{[\ist+1,i_1]} < 0$ (see
    \cref{one-iteration}). Then,
    \begin{equation}
      \bpsi_l \bpsi_l^T (i,j) =
      \frac{1}{i_1 - i_0 + 1}
      \begin{cases}
        \displaystyle \frac{i_1 - \ist}{\ist - i_0 +1} & \text{if} \mspace{8mu} i_0 \le i,j \le
        \ist,\\
        \\
        \displaystyle \frac{\ist - i_0 +1}{i_1 -\ist} & \text{if} \mspace{8mu} \ist+1 \le i,j \le i_1,\\\\
        -1 &\text{if} \mspace{8mu}
        \begin{cases}
          i_0 \le i \le \ist, \mspace{16mu} \ist+1 \le j \le i_1,\\
          \ist+1 \le i \le i_1,  \mspace{16mu} i_0 \le j \le \ist,
        \end{cases}\\
        0 & \text{otherwise.}
      \end{cases}
      \label{psikpsik}
    \end{equation}
  \end{lemma}
  \begin{proof}
    The proof is an elementary calculation based on the definition  of $\bpsi_l$ given by
    \cref{from_l_to_plusone}, and the observation that $\bpsi_1(i) = n^{-1/2}$. Indeed,
    we know from \cref{from_l_to_plusone} that $\bpsi_l$ is piecewise constant, and given by
    \begin{equation}
      \bpsi_l (i) = \frac{1}{\sqrt{i_11 - i_0 +1}}
      \begin{cases}
        \displaystyle \mspace{16mu} \frac{\sqrt{i_1 -\ist}}{\sqrt{\ist - i_0 +1}}   &
        \text{if} \mspace{8mu} i_0 \le i \le \ist,\\ \\
        \displaystyle - \frac{\sqrt{\ist - i_0 +1}}{\sqrt{i_1 -\ist}} & \text{if} \mspace{8mu} \ist+1 \le i \le i_1,\\
        0 & \text{otherwise.}
      \end{cases}
      \label{levecteur_psi_l}
    \end{equation}
    The computation of the tensor product is immediate and yields the advertised result.
  \end{proof}
  \subsubsection{The matrix \texorpdfstring{$\sum_{k=1}^M \bpsi_k \bpsi_k^T$}{TEXT}}
  In the following corollary, we describe the matrix $\sum_{k=1}^M \bpsi_k
  \bpsi_k^T$. When combined with \cref{corollary3p51}, we use this corollary to
  reconstruct the geometry of the blocks in the SBM.
  \begin{corollary}
    \label{corollaryp37}
    Let $\{J_m\}$ be the $M$ leaves in the binary Soules tree (these are intervals that are no
    longer split) after $M$ steps of \cref{BSB}. Then $E_M \eqdef \sum_{k=1}^M
    \bpsi_k \bpsi_k^T$ is equal to
    \begin{equation}
      e_M(i,j) =
      \begin{cases}
        \displaystyle \frac{1}{\lon{J_m}} & \text{if} \mspace{8mu} (i,j) \in J_m \times J_m,\\
        0 & \text{otherwise.}
      \end{cases}
      \label{EM}
    \end{equation}
    Also, 
    \begin{equation}
      \begin{cases}
        \sum_{k=2}^M \bpsi_k \bpsi_k^T (i,j) > 0 & \text{if}\mspace{8mu} \exists m \in
        \{1, 2 \ldots, M\}, \mspace{8mu} (i,j) \in J_m \times J_m,\\
        \sum_{k=2}^M \bpsi_k \bpsi_k^T (i,j) < 0 & \text{otherwise}
      \end{cases}
      \label{EMmoinsPsi1}
    \end{equation}
  \end{corollary}
  \begin{proof}
    We first observe that after $M$ iterations of \cref{BSB} there are $M$ intervals
    $J_m$ that are not split (the leaves in the binary tree shown in
    \cref{the-soules-tree}, where we count the construction of $\bpsi_1$ as the first
    iteration of the algorithm ($M=1$). This can be proved by induction, after observing that
    an iteration of \cref{BSB}, described by \cref{from_l_to_plusone}, turns exactly
    one leaf in the tree into two leaves.\\

    \noindent Next, we prove that $\sum_{k=1}^M \bpsi_k \bpsi_k^T$ is nonnegative on each $J_m
    \times J_m$, $1\le m \le M$. Since, each interval $J_m$ is a leaf of the tree, the
    interval $J_m$ is not further decomposed, and there exists a vector $\bpsi_k$ such that
    $\rstr{\bpsi_k}{J_m} > 0$ or $\rstr{\bpsi_k}{J_m} < 0$ (see
    \cref{the-soules-tree}-right). We can therefore apply \cref{computation-psik},
    with $J_m = [i_0,\ist]$ or $J_m = [\ist, i_1]$, and $\bpsi_k\bpsi_k^T $ is constant on $J_m
    \times J_m$ (see \cref{psikpsik}). All other larger scale (vectors  $\bpsi_k$ such that
    $J_m \subset \supp{\bpsi_k}$, also keep a constant value on $J_m$, and therefore $\bpsi_k
    \bpsi_k^T $ is constant on $J_m \times J_m$. We conclude that $\sum_{k=1}^M \bpsi_k
    \bpsi_k^T$ is constant on each $J_m \times J_m$, $1\le m \le M$.\\

    \noindent We can then prove by induction that
    \begin{equation}
      e_M(i,j) =
      \begin{cases}
        \displaystyle \frac{1}{\lon{J_m}} & \text{if} \mspace{8mu} (i,j) \in J_m \times J_m,\\
        0 & \text{otherwise.}
      \end{cases}
      \label{le-EM}
    \end{equation}
    For $M=1$ there is nothing to prove, since $\bpsi_1 = \scalebox{0.8}{$n^{-1/2}$}
    \one$. Now, assume that \cref{le-EM} holds for $M\ge 1$, then $E_{M+1} = E_M +
    \bpsi_{M+1} \times \bpsi_{M+1}$, and $\bpsi_{M+1}$ is created by splitting an interval
    $J_q, 1 \le q \le M$, such that $\supp{\bpsi_{M+1}} = J_q$, and
    $\rstr{\bpsi_{M+1}}{J_m} = 0$ for all $m \ne q$. Since $J_q$ is the only block that
    changes when going from $M$ to $M+1$, $\sum_{k=1}^{M+1} \bpsi_k \bpsi_k^T$ is equal to
    $\sum_{k=1}^M \bpsi_k \bpsi_k^T$ on all the other blocks. Using the induction
    hypothesis, we then have for all $m\ne q$,
    \begin{equation}
      \forall (i,j) \in J_m \times J_m, \mspace{16mu}
      e_{M+1} (i,j) = e_M(i,j) = \displaystyle \frac{1}{\lon{J_m}}.
    \end{equation}
    We are left with the computation of $\sum_{k=1}^{M+1} \bpsi_k \bpsi_k^T$ on $J_q \times
    J_q$. Let us define $i_0$ and $i_1$ such that $J_q = [i_0,i_1]$, and let $\ist$ be the
    index where $J_q$ is split, $J_q = [i_0,\ist] \cup [\ist+1,i_1]$. Then using
    \cref{computation-psik} we have for all $(i,j) \in J_q\times J_q$,
    \begin{equation}
      \bpsi_{M+1} \times \bpsi_{M+1} (i,j) =
      \frac{1}{i_1 - i_0 + 1}
      \begin{cases}
        \displaystyle \frac{i_1 - \ist}{\ist - i_0 +1} & \text{if} \mspace{8mu} (i,j) \in [i_0, \ist] \times [i_0,\ist],\\
        \displaystyle  \frac{\ist - i_0 +1}{i_1 -\ist} & \text{if} \mspace{8mu} (i,j) \in [\ist+1,i_1] \times [\ist+1,i_1],\\
        -1 &\text{otherwise.}
      \end{cases}
    \end{equation}
    From the induction hypothesis, we have for all  $(i,j) \in J_q\times J_q$, $e_M (i,j) =
    \lon{i_1 - i_0 + 1}^{-1}$. Adding $\sum_{k=1}^M \bpsi_k \bpsi_k^T$ and $\bpsi_{M+1} \bpsi_{M+1}^T$ yields
    for all  $(i,j)
    \in J_q\times J_q$,
    \begin{equation}
      e_{M+1} (i,j) =
      \begin{cases}
        \displaystyle \frac{1}{\ist - i_0 +1} & \text{if} \mspace{8mu} (i,j) \in [i_0, \ist] \times [i_0,\ist],\\
        \displaystyle \frac{1}{i_1 -\ist} & \text{if} \mspace{8mu} (i,j) \in [\ist+1,i_1] \times [\ist+1,i_1],\\
        0 &\text{otherwise,}
      \end{cases}
    \end{equation}
    which concludes the case for $M+1$. By induction, \cref{le-EM} holds for all $M$.\\

    \noindent We conclude the proof of \cref{corollaryp37} by proving
    \cref{EMmoinsPsi1}. Let $i,j$ be two nodes in the leaf $J_q$ then $e_M(i,j) =
    \lon{J_q}^{-1}$. Also, $\bpsi_1 \times \bpsi_1 (i,j) = n^{-1}$, and thus
    \begin{equation}
      \sum_{m=2}^M \bpsi_k \bpsi_k^T (i,j)  = \frac{1}{\lon{J_q}} - \frac{1}{n} >  0,
    \end{equation}
    since $\lon{J_q} > 1$. Now,  if $(i,j)$ is not in any blocks $J_q \times J_q$, then
    $e_M(i,j) = 0$, and therefore $e_M(i,j) - \bpsi_1 \times \bpsi_1 (i,j) = -n^{-1} <
    0$.
  \end{proof}
  We now prove a series of lemmata that address the performance of \cref{BSB} and its
  ability to detect the blocks of a general $\sbm{\bp}{q}{n}$ by aligning the successive
  $\psi_m$ with the block boundaries. The proof hinges on the study of one iteration of
  \cref{BSB}, as explained in \cref{maximum-coefficient}.
  \subsubsection{One iteration of \cref{BSB}}
  The next lemma studies a single iteration of \cref{BSB}, which leads to
  the construction of the Soules vector $\bpsi_l$. We assume that $\supp{\bpsi_l} =
  [i_0,i_1]$, and we consider the matrix $\bP$ that is nonzero only on $[i_0,i_1] \times
  [i_0,i_1]$, and is piecewise constant on two blocks $J_0 \times J_0$ and $J_1 \times J_1$,
  where $J_0 = [i_0,j]$, and $J_1 = [j+1,i_1]$ (see \cref{2blocks}),
  \begin{equation}
    \bP
    = p_0 \big(\one_{J_0}\one_{J_0}^T\big)
    + p_1 \big(\one_{J_1}\one_{J_1}^T \big) 
    + q\big( \one_{J_0}\one_{J_1}^T +\one_{J_1} \one_{J_0}^T \big).
  \end{equation}
  \noindent We prove that in order to maximize $\lvert \ipr{\bpsi_l
    \bpsi_l^T}{\bP}\rvert^2$, \cref{BSB} must always align $\ist$ (the zero-crossing
  of $\bpsi_l$) with the jump in the SBM inside $\supp{\bpsi_l\bpsi_l^T}$ (see
  \cref{2blocks}). 
  \begin{lemma}
    \label{maximum-coefficient}
    Let $\bpsi_l$ be the Soules vector returned by \cref{BSB} at level $l$ with
    support $\supp{\bpsi_l} \eqdef [i_0,i_1]$. We consider the matrix $\bP$ that is nonzero only on
    $[i_0,i_1] \times [i_0,i_1]$, and is piecewise constant on two blocks $J_0 \times J_0$
    and $J_1 \times J_1$, where $J_0 = [i_0,j]$, and $J_1 = [j+1,i_1]$ (see
    \cref{2blocks}),
    \begin{equation}
      \bP
      = p_0 \big(\one_{J_0}\one_{J_0}^T\big)
      + p_1 \big(\one_{J_1}\one_{J_1}^T \big) 
      + q\big( \one_{J_0}\one_{J_1}^T +\one_{J_1} \one_{J_0}^T \big).
    \end{equation}
    Then, $\lvert \ipr{\bpsi_l \bpsi_l^T}{\bP}\rvert^2$ is maximum if the location of the
    zero-crossing of $\bpsi_l$ is equal to the location of the jump in the SBM, $\ist=j$ (see
    \cref{2blocks}).
  \end{lemma}
  \begin{proof}
    The proof relies on the computation of the inner-product between a Soules tensor product
    $\bpsi_l \bpsi_l^T$ and an SBM whose support coincide with the support of $\bpsi_l
    \bpsi_l^T$. We use \cref{computation-psik}, and we study two cases for the choice of
    $\ist \in [i_0,i_1]$. We have
    \begin{equation}
      \ipr{\bpsi_l \bpsi_l^T}{\bP} = 
      p_0  \ipr{\bpsi_l \bpsi_l^T}{\one_{J_0} \one_{J_0}}^T +
      p_1  \ipr{\bpsi_l \bpsi_l^T}{\one_{J_1}\one_{J_1}^T} 
      + c  \ipr {\bpsi_l \bpsi_l^T}{\one_{J_0}\one_{J_1}^T +\one_{J_1}\one_{J_0}^T}.
    \end{equation}
    \begin{figure}[H]
      \centerline{
        \includegraphics[width=0.45\textwidth]{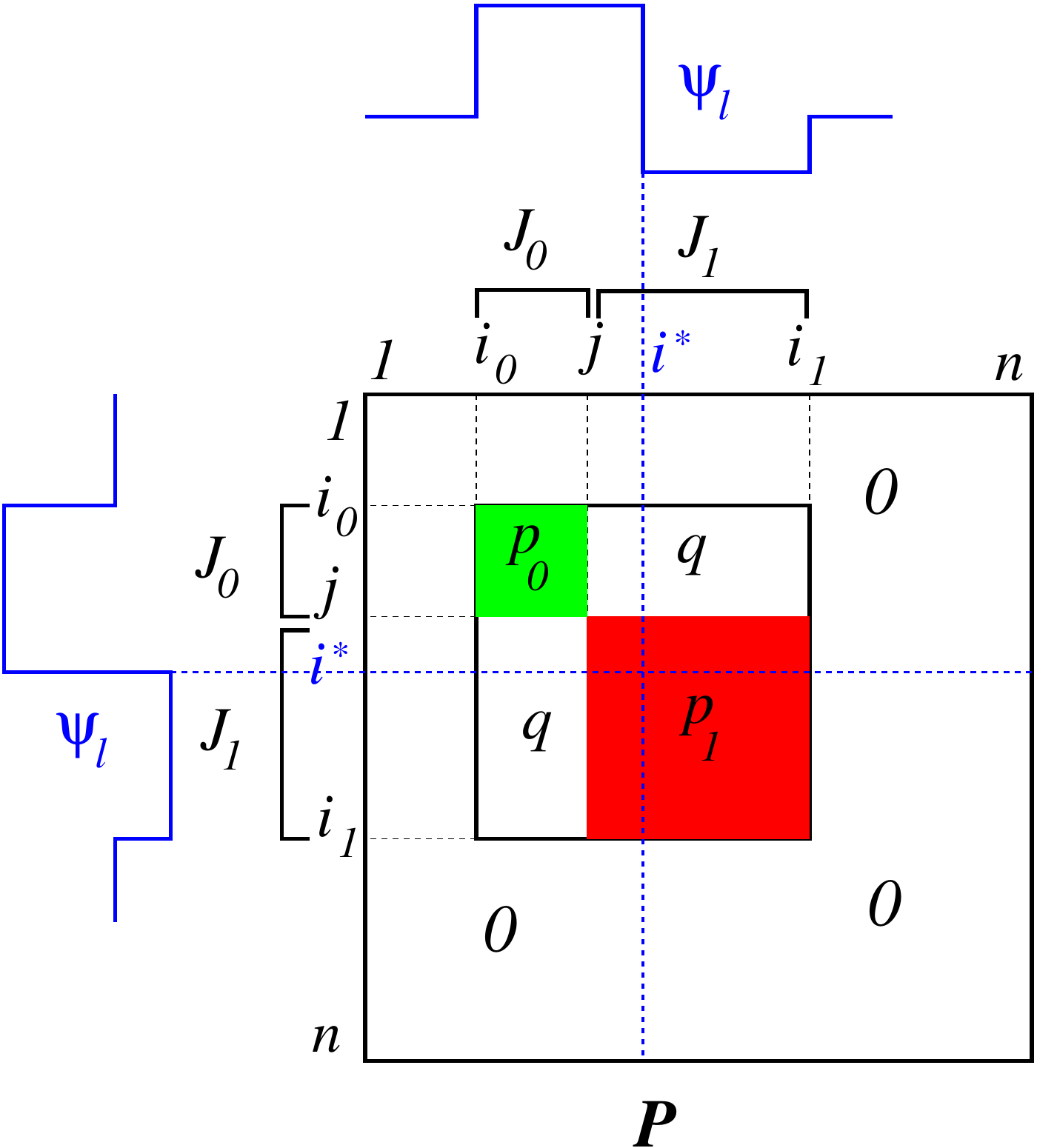}
      }
      \caption{The vector $\bpsi_l$ (in blue) is created by splitting a block of indices
        $[i_0,i_1]$ at level $l-1$ into two sub-blocks, $[i_0,\ist] \cup [\ist+1,i_1]$ at level
        $l$.  We consider the matrix $\bP$ that is nonzero only on $[i_0,i_1] \times
        [i_0,i_1]$, and is piecewise constant on two blocks $J_0 \times J_0$ (in green) and $J_1
        \times J_1$ (in red), where $J_0 = [i_0,j]$, and $J_1 = [j+1,i_1]$
        \label{2blocks}}
    \end{figure}
    \noindent    Also, $\ipr {\bpsi_l \bpsi_l^T}{\one_{J_q}\one_{J_r}^T} = \ipr{\bpsi_l}{\one_{J_q}}
    \ipr{\bpsi_l}{\one_{J_r}},\mspace{8mu} \text{for} \mspace{4mu} q,r \in \{0,1\}$.  We
    define $r_q \eqdef \ipr{\bpsi_l}{\one_{J_q}}$ for $q=0,1$. Then
    \begin{equation}
      \ipr{\bpsi_l =  \bpsi_l^T}{\bP}  = p_0 r_0^2 + 2q r_0r_1 + p_1 r_1^2.
    \end{equation}
    The expression of the coefficients $r_0$ and $r_1$ can be derived by using
    \cref{levecteur_psi_l}. We give the details for the computation of $r_0$, the computation
    of $r_1$ is very similar. To compute $r_0$, we need to consider the two cases, $i_0 \le \ist
    \le j$ and $j \le \ist \le i_1$. We recall from \cref{levecteur_psi_l} that we always have
    \begin{equation}
      \bpsi_l (i) = \frac{1}{\sqrt{i_1 - i_0 +1}}
      \begin{cases}
        \displaystyle \mspace{16mu} \frac{\sqrt{i_1 -\ist}}{\sqrt{\ist - i_0 +1}}   &
        \text{if} \mspace{8mu} i_0 \le i \le \ist,\\ \\
        \displaystyle - \frac{\sqrt{\ist - i_0 +1}}{\sqrt{i_1 -\ist}} & \text{if} \mspace{8mu} \ist+1 \le i \le i_1,\\
        0 & \text{otherwise.}
      \end{cases}
    \end{equation}
    \noindent If $\ist \le j$ then $\bpsi_l$ changes sign over $J_0$ and we have
    \begin{equation}
      \begin{aligned}
        r_0 & = \ipr{\bpsi_l}{\one_{J_0}}
        = \frac{1}{\sqrt{i_1 - i_0 +1}} \bigg\{\sum_{i=i_0}^\ist \frac{\sqrt{i_1 -\ist}}{\sqrt{\ist - i_0 +1}}
        - \sum_{i=\ist+1}^j \frac{\sqrt{\ist - i_0 +1}}{\sqrt{i_1 -\ist}} \bigg \}\\
        & = \sqrt{\frac{(\ist -i_0 +1)(i_1 -\ist)}{i_1 - i_0 +1}}\bigg(\frac{i_1 - j}{i_1 - \ist}\bigg).
      \end{aligned}
    \end{equation}
    \noindent If $j \le \ist$ then $\bpsi_l$ is positive over $J_0$ (this is the case for
    \cref{2blocks}) and we have
    \begin{equation}
      r_0 = \ipr{\bpsi_l}{\one_{J_0}} 
      = \frac{1}{\sqrt{i_1 - i_0 +1}} \sum_{i=i_0}^j \frac{\sqrt{i_1 -\ist}}{\sqrt{\ist - i_0 +1}}
      = \sqrt{\frac{(\ist -i_0 +1)(i_1 -\ist)}{i_1 - i_0 +1}}\bigg(\frac{j - i_0 + 1}{\ist - i_0 +1}\bigg).
    \end{equation}
    A similar calculation yields $r_1$. If $\ist +1\le j+1$ then $\bpsi_l$ is negative  over $J_1$ and we have
    \begin{equation}
      r_1 = - \sqrt{\frac{(\ist -i_0 +1)(i_1 -\ist)}{i_1 - i_0 +1}}\bigg(\frac{i_1 - j}{i_1 - \ist}\bigg),
    \end{equation}
    and if $j+1 \le \ist +1$ (this is the case for \cref{2blocks}) then $\bpsi_l$ changes sign
    over $J_1$ and we have
    \begin{equation}
      r_1 = - \sqrt{\frac{(\ist -i_0 +1)(i_1 -\ist)}{i_1 - i_0 +1}}\bigg(\frac{j - i_0 + 1}{\ist - i_0
        +1}\bigg).
    \end{equation}
    We are now ready to evaluate $\ipr{\bpsi_l\bpsi_l^T}{\bP} = p_0 r_0^2 + 2q r_0r_1 + p_1
    r_1^2$. Again, we need to consider the following two cases. If $\ist \le j$ then
    \begin{equation}
      \begin{aligned}
        \ipr{\bpsi_l\bpsi_l^T}{\bP}
        & =  p_0 \frac{(\ist -i_0 +1)(i_1 -\ist)}{i_1 - i_0 +1}\bigg(\frac{i_1 - j}{i_1 - \ist}\bigg)^2
        +  p_1 \frac{(\ist -i_0 +1)(i_1 -\ist)}{i_1 - i_0 +1}\bigg(\frac{i_1 - j}{i_1 - \ist}\bigg)^2\\
        & - 2q \frac{(\ist -i_0 +1)(i_1 -\ist)}{i_1 - i_0 +1}\bigg(\frac{i_1 - j}{i_1 - \ist}\bigg)^2\\
        & = \frac{(\ist -i_0 +1)(i_1 -\ist)}{i_1 - i_0 +1}\bigg(\frac{i_1 - j}{i_1 - \ist}\bigg)^2
        \big\{p_0 + p_1 -2q\big\},
      \end{aligned}
    \end{equation}
    which is maximum when $\ist = j$. In the case where  if $j \le \ist$ we have
    \begin{equation}
      \ipr{\bpsi_l\bpsi_l^T}{\bP}
      = \frac{(\ist -i_0 +1)(i_1 -\ist)}{i_1 - i_0 +1}
      \bigg(\frac{j - i_0 + 1}{\ist - i_0 + 1}\bigg)^2
      \big\{p_0 + p_1 -2q\big\},
    \end{equation}
    which is also maximum when $\ist = j$. This concludes the proof that
    $\ipr{\bpsi_l\bpsi_l^T}{\bP}$ is maximal if $\ist=j$.
  \end{proof}
  \noindent \cref{first-cut} extends \cref{maximum-coefficient} to the general
  edge probability matrix $\bP$ of an SBM (see \cref{SBM})); it is used to prove
  \cref{corollary3p51} by induction. \cref{first-cut} can be proved using a proof
  by contradiction (using \cref{maximum-coefficient}).

  \begin{lemma}
    \label{first-cut}
    Let $\bP$ be the population mean adjacency matrix of $\sbm{\bp}{q}{n}$ defined by
    \begin{equation}
      \bP = \sum_{m=1}^M (p_m - q)\one_{B_m}\one_{B_m}^T + q \bJ,
    \end{equation}
    where the $M$ blocks $\{B_m\}$ form a partition of $[n]$.  Then, the split that creates
    $\bpsi_2$ in \cref{BSB}, is always located at the boundary between two blocks $B_m$
    and $B_{m+1}$.
  \end{lemma}
  \begin{proof}
    Let $\ist$ be the index associated with the construction of $\bpsi_2$ and the subdivision
    of $[n]$. We need to prove that $\ist$ coincides with the endpoint of a block $B_m$.  By
    contradiction, if $\ist$ does not correspond to the boundary between two blocks, then there
    exists $i_0< i_1$ such that $B_m=[i_0,i_1]$ and $i_0 < \ist < i_1$. Since $\bP$ is constant
    over the block $[i_0,i_1] \times [i_0,i_1]$ (see \cref{2blocks} with $p_0 = p_1 =
    q$), \cref{maximum-coefficient} tells us that the value of $\bP$ in $B_m \times
    B_m$ does not contribute to $\lvert \ipr{\bpsi_2 \bpsi_2^T}{\bP}\rvert^2$, and
    \cref{BSB} should not have placed $\ist$ in $B_m$.
  \end{proof}
  \subsubsection{$M$ iterations of \cref{BSB}}
  This last lemma guarantees that after $M$ iterations of \cref{BSB}, the matrix
  $\sum_{k=1}^M \bpsi_k \bpsi_k^T$ associated with the first $M$ Soules vectors recovers the block
  geometry. \cref{corollary3p51} is proved by induction on $M$, using \cref{first-cut}.
  \begin{lemma}
    \label{corollary3p51}
    Let $\bP$ be the population mean adjacency matrix of $\sbm{\bp}{q}{n}$ defined by
    \begin{equation}
      \bP = \sum_{m=1}^M (p_m - q)\one_{B_m}\one_{B_m}^T + q \bJ.
      \label{unSBM}
    \end{equation}
    Let $J_l, 1 \le l \le M$ be the leaves in the binary Soules tree (these are intervals
    that are no longer split) after $M$ steps of \cref{BSB}. Then, the $M$ blocks
    $\{B_m\}$ in \cref{unSBM} coincide with the $M$ intervals $\{J_l\}$ discovered by
    \cref{BSB}.
  \end{lemma}
  \begin{proof}
    We prove the result by induction on $M$. If $M=1$, there is nothing to prove. If $M=2$,
    then \cref{first-cut} shows that $\bpsi_2$ recovers the block geometry. We assume that
    the result holds for all matrices $\bP$ with $m \le M$ blocks that are given by
    \cref{unSBM}. We consider the population mean adjacency matrix $\bQ$ defined by
    \begin{equation}
      \bQ = \sum_{m=1}^{M+1} (p_m - q)\one_{C_m}\one_{C_m}^T + q \bJ,
    \end{equation}
    where $\cup_{m=1}^{M+1} C_m = [n]$. Because of \cref{first-cut}, the first split of
    $[n]$, which leads to the construction of $\bpsi_2$ is aligned with the boundary of a
    block $C_{m_0}= [i_0,\ist]$. Without loss of generality, we can assume that the cut is
    aligned with the endpoint of $C_{m_0}$. We can then partition $\bQ = \bQ_1 + \bQ_2$,
    where
    \begin{equation}
      \bQ_1 = \sum_{m=1}^{\ist} (p_m - q)\one_{C_m}\one_{C_m}^T + q \one_{[\ist]}\one_{[\ist]}^T
    \end{equation}
    and
    \begin{equation}
      \bQ_2 = \sum_{m=\ist+ 1}^{M+1} (p_m - q)\one_{C_m}\one_{C_m}^T + q
      \one_{\{\ist+1,\ldots,n\}}\one_{\{\ist+1,\ldots,n\}}^T.
    \end{equation}
    Again, because of \cref{first-cut}, the next splits happen (independently) in
    $\bQ_1$, or $\bQ_2$. We can use the induction hypothesis to argue that all further
    splits will be located along the blocks in $\bQ_1$, or $\bQ_2$. After $M$ splits, the
    algorithm has detected all $M+1$ blocks. By induction, the result holds for all $M$.
  \end{proof}
  \noindent \cref{lemma4} is then a direct consequence of \cref{corollary3p51} and \cref{corollaryp37}.


\end{document}